\documentclass[final,3p,times]{elsarticle}

\usepackage{amssymb}

\usepackage[cmex10]{amsmath}
\usepackage{amsfonts}

\usepackage{subfig}

\usepackage{bbm}
\usepackage[numbers]{natbib}
\usepackage{graphicx}
\usepackage{epstopdf}
\usepackage{url}

\usepackage{lineno}
\usepackage{color}
\usepackage{amsthm}

\journal{Applied Mathematics and Computation}

\newtheorem{thm}{Theorem}[section]

\newtheorem{definition}{Definition}

\begin{document}

\begin{frontmatter}

\title{How to Run a Campaign: Optimal Control of SIS and SIR Information Epidemics}
\author{Kundan Kandhway\corref{cor1}}
\ead{kundan@dese.iisc.ernet.in}
\cortext[cor1]{Corresponding Author}
\author{Joy Kuri}
\ead{kuri@dese.iisc.ernet.in}
\address{Department of Electronic Systems Engineering, Indian Institute of Science, Bangalore 560012, India.}

\begin{abstract}
Information spreading in a population can be modeled as an epidemic. Campaigners (\emph{e.g.} election campaign managers, companies marketing products or movies) are interested in spreading a message by a given deadline, using limited resources. In this paper, we formulate the above situation as an optimal control problem and the solution (using Pontryagin's Maximum Principle) prescribes an optimal resource allocation over the time of the campaign. We consider two different scenarios---in the first, the campaigner can adjust a direct control (over time) which allows her to recruit individuals from the population (at some cost) to act as spreaders for the Susceptible-Infected-Susceptible (SIS) epidemic model. In the second case, we allow the campaigner to adjust the effective spreading rate by incentivizing the infected in the Susceptible-Infected-Recovered (SIR) model, in addition to the direct recruitment. We consider time varying information spreading rate in our formulation to model the changing interest level of individuals in the campaign, as the deadline is reached. In both the cases, we show the existence of a solution and its uniqueness for sufficiently small campaign deadlines. For the fixed spreading rate, we show the effectiveness of the optimal control strategy against the constant control strategy, a heuristic control strategy and no control. We show the sensitivity of the optimal control to the spreading rate profile when it is time varying.
\end{abstract}


\end{frontmatter}

\section{Introduction}

Use of social networks by political campaigners and product marketing managers is increasing day by day. It gives them an opportunity to influence a large population connected via the online network, as well as the human network where two individuals interacting with each other in daily life are connected. A piece of information, awareness of brands, products, ideas and political ideologies of candidates spreads through such a network much like pathogens in the human network, and this phenomenon is called an ``information epidemics''. The goal of the campaigner is to `infect' as many individuals or nodes as possible with the message by the campaign deadline. Such an effort incurs advertisement cost. Some campaigns aim to create a `buzz' about some topic by engaging people in a conversation about some topic (a scenario encountered in political campaigns, for example). Others are more focused: for example, advertisements to maximize the sale of a product or promotion of a movie. Resource limitations (monetary, manpower or otherwise) indicate the need to formulate optimal campaigning strategies which can achieve these goals at minimum cost. Such an optimization problem can be formulated as an optimal control problem. An optimal control problem aims to optimize a cost functional (a function of state and control variables) subject to state equation constraints that govern system evolution.

In this paper, we aim to address the problems described above. We assume a homogeneously mixed population. Individuals communicate and exchange messages with one another on their own, giving rise to information epidemics. Information can be communicated to the population directly by the campaigner (direct recruitment of individuals to spread the message). However, recruiting individuals and direct information communication comes with a cost (such as placing advertisement in the mass media). A campaigner may also provide incentives to individuals for spreading the message. Such an incentive is termed as a word-of-mouth incentive and it rewards an individual who refers a product or a piece of information to others.

The campaigner possesses limited resources and is unable to communicate information to the entire population. Not only resource allocation among different strategies, but also the timing of direct recruitment of individuals and giving out word-of-mouth incentives are crucial for maximizing the information epidemic. We model the information spreading process as a Susceptible-Infected-Susceptible (SIS) and Susceptible-Infected-Recovered (SIR) epidemic process with time varying information spreading rate, and formulate an optimal control problem which aims to minimize the campaign cost over a given period of time.

SIS and SIR epidemic processes are suitable for modeling information epidemics due to the similarities in the ways disease spreads in a biological network and this information spreads in social networks. When susceptible and infected individuals interact, the topic of interest may come up with some probability, which will lead to transfer of information from infected to susceptible individual. This process is very similar to the way in which a communicable disease spreads in a population. The SIS model is suitable for cases when we are trying to engage the population in a conversation about some topic. Such a scenario can be encountered in political campaigns. SIS allows infected nodes to `recover' back to the susceptible state, so that it can receive a different message about the same topic. The SIR model is suitable for situations where nodes participate in message spreading for random amounts of time and then recover (and stop message dissemination). Such a scenario may be encountered in viral marketing of a newly launched product or promotion of a movie, where information about the product is transmitted by enthusiastic individuals who gradually lose interest in promoting the product.

\textbf{\emph{Related Work:}} Although there are a lot of studies on \emph{preventing} the spread of disease and computer viruses in human and computer networks through optimal control \citep{asano2008optimal,behncke2000,gaff2009optimal,lashari2012optimal,ledzewicz2011optimal,morton1974,sethi1978,yan2008,zhu2012optimal,castilho2006optimal}, information epidemics have attracted less attention (see for example \cite{karnik2012}). Apart from differences in epidemic models, our objective is to maximize the spread of information while the studies discussed above aim to contain the epidemic. To be more precise, the above studies aim to remove individuals from the infected class through the application of control signal(s), while our aim is the opposite. In addition, the cost functions used by \cite{behncke2000,morton1974} and \cite{sethi1978} are linear in control, while our cost is quadratic in control. \cite{yan2008} and \cite{zhu2012optimal} assumed a time-varying state variable in the cost function, while our formulation considers only the final system state. \cite{asano2008optimal} considered a metapopulation model of epidemics which is different from the models used in this paper. The SIR model is used in \cite{gaff2009optimal} and \cite{ledzewicz2011optimal}, but the controls (level of vaccination and treatment) are specific to biological epidemics and unsuitable for information epidemics. The work in \cite{castilho2006optimal} aims to prevent the epidemic and the author uses educational campaigns as controls. The educational campaigns encourage susceptibles to protect themselves from the disease and increase the removal rate of infected individuals.

\citep{chierichetti2009} and \citep{pittel1987spreading} analyze the so called Push, Pull and Push-pull algorithms for rumor spreading on technological and social networks. The connection graph of the nodes is known. The authors fix the strategy for information spreading: nodes may either `Push' the information to their neighbors, `Pull' the information from them or do both. The aim is to compute bounds on the number of communication rounds required to distribute the message to almost all the nodes for the given connection graph. Similarly, the authors in \cite{boyd2005} aim to compute bounds on the number of communication rounds required for distributed computation of the average value of sensor readings (\emph{e.g.} temperature) for sensors deployed in a field. Note that finding optimal strategies for information spreading is not the aim of \citep{chierichetti2009,pittel1987spreading,boyd2005}.

The study in \cite{belen2008} defines an information or a joke spreading in a population as a rumor. It aims to maximize the spread of rumor in the Daley-Kendall and Maki-Thompson models, which are different from the Kermack-McKendrick SIS/SIR models used in this paper. We believe that certain scenarios like political campaigns are better modeled by Kermack-McKendrick SIS model (recovery is independent of interaction between individuals) than the Daley-Kendall/Maki-Thompson models (when two infected meet, one or both recover). Moreover, \cite{belen2008} used impulse control, while our models assume the system can be controlled over the whole campaign period. The Authors in \cite{sethi2008optimal} and \cite{krishnamoorthy2010optimal} devise optimal advertisement and pricing strategies for newly launched products. But they do not consider viral message propagation, where individuals in the population interact with one another to spread a piece of information, as is the case in this paper.

Optimal control of spreading software security patches in technological networks is discussed in \cite{khouzani2011optimal}. The system model used by \cite{khouzani2011optimal} is tailored to technological networks and is different from the one used in our study. The study uses a four-compartment model (instead of the two and three compartment models used in our study) with an additional compartment for representing the number of nodes destroyed by the malware (apart from susceptibles, infected and recovered nodes). In addition, the controls used are the fraction of disseminators and the dissemination rate of the security patches, which are different from the controls used in this paper.

Cost minimization for marketing campaigns is explored in \cite{karnik2012}. The cost functional used there is linear in control, while our model assumes a quadratic cost. Although \cite{karnik2012} uses SIS and SIR models like we do, but the objective functional is different from ours. The objective functionals in our models are weighted sums of the fraction of nodes which have received the message by the campaign deadline and the costs of applying controls. In contrast, the objective functional in \cite{karnik2012} for the SIS case is to minimize the campaign cost (subject to a constraint on the fraction of infected nodes at the deadline) and for the SIR case, the objective is to maximize the fraction of nodes who have received the message (subject to a budget constraint for running the campaign). Moreover, we discuss the uniqueness of the solutions to our models; this is not discussed by \cite{karnik2012} and \cite{khouzani2011optimal}. Furthermore, we study the effect of word-of-mouth control on epidemic spreading, which is not explored in the above studies.

\textbf{\emph{A major difference between an information epidemic and a biological epidemic}} is that in the case of a biological epidemic, the infection rate and recovery rate are constant throughout the season (assuming that pathogens do not mutate within a season). On the other hand, the interest level of the population during the campaign period (for elections or promotions of upcoming movies) changes as we approach the deadline (poll date or movie release date). \emph{We have modeled this by making the effective information spreading rate a time varying quantity.} Previous studies have ignored this characteristic of information epidemics.

\textbf{\emph{The following are the primary contributions of this paper:}}
\begin{enumerate}[(i)]
\item Formulation of optimal control problems for maximizing information spread in two different models. In the first, information spreads through an SIS process and in the next, through an SIR process. The control signal in the SIS model is the intensity with which direct recruitment of spreaders from the population can be done. The SIR model has an additional control signal---the word-of-mouth control, which controls the spreading rate of the information epidemic. Our formulation involves a quadratic cost function and a time dependent effective information spreading rate.
\item We show the existence of a solution in both the problems.
\item We establish uniqueness of the solution in both the cases, when the campaign deadline is sufficiently small.
\item We compare the overall cost incurred by the optimal control strategy with the constant control strategy, a heuristic control strategy and no control, when the information spreading rate is constant throughout the campaign period.
\item When the effective spreading rate varies over the campaign duration, we demonstrate the sensitivity of the optimal control with respect to different effective spreading rate profiles.
\end{enumerate}

The rest of the paper is organized as follows: Sections \ref{sec:SIS_sys_model} and \ref{sec:SIR_sys_model} formulate the optimal control problem for SIS and SIR information epidemics, respectively. Sections \ref{sec:SIS_model_analysis} and \ref{sec:SIR_model_analysis} analyze the respective problems. Results are shown in Section \ref{sec:Results} and conclusions are drawn in Section \ref{sec:Conclusions}.

\section{System Model and Problem Formulation: SIS Epidemic}
\label{sec:SIS_sys_model}
We consider a system of $N$ nodes (or individuals) which is fixed throughout the campaign time $0\leq t\leq T$. In the case of SIS epidemics, individuals are divided into two compartments---susceptible (those who are yet to receive a tagged message) and infected (those who have already received the message). A susceptible node becomes infected at a certain rate if it comes in contact with an infected individual. The idea is to create a `buzz' in the population about some topic (\emph{e.g.} a political campaign). An infected node can `recover' back to the susceptible state which allows it to receive a different message about the same topic. This makes SIS a suitable model for such scenarios. By the campaign deadline $T$, the number of infected individuals (who are engaged in some sort of discussion about the topic) is the quantity of interest.

Let $S(t)$, $I(t)$ denote the number of susceptible and infected nodes at time $t$. Let $s(t)=S(t)/N \ge 0$ and $i(t)=I(t)/N \ge 0$, therefore $s(t)+i(t)=1$. The information spreading rate at time $t$ is denoted by $\beta'(t)$, which is assumed to be a bounded quantity. Practical considerations will impose such a restriction on $\beta'(t)$. In a small interval $dt$ at time $t$, a susceptible node that is in contact with a single infected node, changes its state to ``infected'' with probability $\beta'(t) dt$. We assume that each node in the population is in contact with an average of $k_m$ others, chosen randomly, at any time instant $t$. Thus, any susceptible node will have an average of $k_m i(t)$ infected neighbors and will acquire information with probability $1-(1-\beta'(t)dt)^{k_mi(t)} \simeq \beta'(t) k_m i(t) dt$. Since a fraction $s(t)$ of population is susceptible, the rate of increase of infected nodes in the population due to susceptible-infected contact is $\beta'(t) k_m i(t) s(t)$. We define the effective spreading rate at time $t$ as $\beta(t) = \beta'(t) k_m$. An infected node ``falls back'' to being susceptible with a rate $\gamma$. In the limit of large $N$, the mean field equations for the evolution of $s(t)$ and $i(t)$ in the (uncontrolled) SIS process is given by \cite[equations adapted for time varying $\beta(t)$]{barrat2008dynamical},
\begin{eqnarray}
\dot{s}(t) &=& -\beta(t) s(t) i(t) + \gamma i(t),  \nonumber \\ \dot{i}(t) &=& \beta(t) s(t) i(t) - \gamma i(t).   \nonumber
\end{eqnarray}

The objective functional is chosen to be $J = -i(T) + \int_{0}^{T} b u^2(t)dt$. The rationale behind such a choice is as follows. Applications like poll campaigns only care about the final number of infected individuals on the polling day, $i(T)$, and not on the evolution history, $i(t), 0\leq t <T$. This is captured in the first term of the objective functional. The cost of running the campaign accrues over time and this is represented by the integral. The choice of quadratic cost function is consistent with the literature \citep{gaff2009optimal,zhu2012optimal}.

Let $u$---a bounded Lebesgue integrable function---denote the recruitment control applied by the campaigner, with $u(t)$ representing its value at time $t$. We define $U$ as follows: 

\begin{definition}
\label{def:SIS_U}
\begin{eqnarray}
u\in U& \triangleq &\{u:u~\textnormal{is Lebesgue integrable},~0\leq u(t)\leq u_{max}\}. \nonumber
\end{eqnarray}
\end{definition}
Thus, $u_{max}$ uniformly bounds all functions in $U$. The control signal $u$ denotes the rate at which nodes are directly recruited from the population to act as spreaders. Practical constraints on executing the control will impose the property of boundedness on the control signal. Thus the optimal control problem can be formulated as:
\begin{eqnarray}
\underset{u\in U}{\text{min}}~~ J &=& - i(T) + \int_{0}^{T} b u^2(t)dt \label{eq:SIS_prob} \\
\text{subject to:} \ \ \dot{s}(t) &=& -\beta(t) s(t) i(t) + \gamma i(t) - u(t)s(t)  \nonumber \\
\dot{i}(t) &=& \beta(t) s(t) i(t) - \gamma i(t) + u(t)s(t)  \\ 
i(t) &\geq& 0, \  \ s(t) \geq 0   \nonumber \\
i(t) &+& s(t) \ = \ 1 \nonumber \\
i(0) &=& i_0, \ \ s(0) = 1 - i_0. \nonumber
\end{eqnarray}
Here, $i_0$ denotes the initial fraction of infected nodes who act as seeds of the epidemic.

\section{System Model and Problem Formulation: SIR Epidemics}
\label{sec:SIR_sys_model}

An SIR epidemic model has an additional compartment---recovered---in addition to the susceptible and infected classes discussed before. It is suitable in modeling situations where nodes participate in message spreading for a random amount of time and then ``recover'' (and stop message dissemination). Such a scenario may be encountered in viral marketing of a newly launched product or promotion of a movie, where enthusiastic individuals gradually lose interest in promoting the product.

Let $R(t)$ and $r(t)=R(t)/N\geq 0$ denote the number and fraction of recovered nodes at time $t$ so that $s(t)+i(t)+r(t)=1$. The effective information spreading rate at time $t$ is $\beta(t)$. Simultaneously, infected nodes switch to ``recovered" at a rate $\gamma$, independent of others. The mean field equations governing the SIR process in the limit of large $N$ are \cite[adapted for variable $\beta(t)$]{barrat2008dynamical}:
\begin{eqnarray}
\dot{s}(t) &=& -\beta(t) s(t) i(t) \nonumber \\
\dot{i}(t) &=& \beta(t) s(t) i(t) - \gamma i(t) \nonumber \\
\dot{r}(t) &=& \gamma i(t).  \nonumber
\end{eqnarray}

In this case we assume that the campaigner can allocate her resources in two ways. At time $t$, she can directly recruit individuals from the population with rate $u_1(t)$, to act as spreaders (via advertisements in mass media). In addition, she can incentivize infected individuals to make further recruitments (\emph{e.g.} monetary benefits, discounts or coupons to current customers who refer their friends to buy services/products from the company). This effectively increases the spreading rate of the message at time $t$ from $\beta(t)$ to $\big(\beta(t)+u_2(t)\big)$ where $u_2(t)$ denotes the ``word-of-mouth'' control signal which the campaigner can adjust at time $t$. The controls, $u_1 \in U_1$ and $u_2\in U_2$, where the sets are defined as follows:
\begin{definition}
\label{def:SIR_U1_U2}
\begin{eqnarray}
U_1&\triangleq&\{u:u~\textnormal{is Lebesgue integrable},~0\leq u(t)\leq u_{1max}\}, \nonumber \\
U_2&\triangleq&\{u:u~\textnormal{is Lebesgue integrable},~0\leq u(t)\leq u_{2max}\}. \nonumber
\end{eqnarray}
\end{definition}
The cost of applying the control is quadratic over the time horizon of the campaign, $0\leq t\leq T$, and the reward is the total fraction of population which received the message at some point in time, i.e., $i(T)+r(T)=1-s(T)$. The optimal control problem can then be stated as:
\begin{eqnarray}
\underset{u_1\in U_1,u_2\in U_2}{\text{min }} J &=& -1 + s(T) + \int_{0}^{T} \big(bu_1^2(t)+cu_2^2(t)\big)dt \label{eq:prob_SIR_word_of_mouth} \\
\text{subject to} ~~ \dot{s}(t) &=& -\big(\beta(t) +  u_2(t)\big)s(t) i(t) - u_1(t)s(t)  \nonumber \\
\dot{i}(t) &=& \big(\beta(t) + u_2(t)\big)s(t) i(t) + u_1(t)s(t) - \gamma i(t)    \nonumber \\
\dot{r}(t) &=& \gamma i(t)  \nonumber \\
i(t) &\geq& 0,~~ s(t) \geq 0, ~~ r(t) \geq 0   \nonumber \\
i(t) &+& s(t)~ +~ r(t)~ =~ 1 \nonumber \\
i(0) &=& i_0,~~ s(0) = 1 - i_0,~~ r(0) = 0. \nonumber
\end{eqnarray}

\section{Analysis of the Controlled SIS Epidemic}
\label{sec:SIS_model_analysis}
Substituting $s(t)=1-i(t)$, problem (\ref{eq:SIS_prob}) can be rewritten as,
\begin{eqnarray}
\underset{u \in U}{\text{min}}~~ J &=& -i(T) + \int_{0}^{T} b u^2(t)dt \label{eq:SIS_prob_reduced} \\
\text{subject to} ~~ \dot{i}(t) &=& - \beta(t) i^2(t) + \big(\beta(t) - \gamma - u(t)\big)i(t) + u(t) \label{eq:SIS_constraint} \\
0  &\leq&  i(t) \leq 1 \ \nonumber \\
i(0) &=& i_0 \label{eq:SIS_init_cond}.
\end{eqnarray}

\subsection{Existence of a Solution} 
\label{sec:ExistenceSIS}

\begin{thm}
There exist an optimal control signal $u\in U$ and a corresponding solution $i^*(t)$ to the initial value problem (\ref{eq:SIS_constraint}) and (\ref{eq:SIS_init_cond}) such that $u\in\underset{u\in U}{\textnormal{argmin }}\{J(u)\}$ in problem (\ref{eq:SIS_prob_reduced}).
\end{thm}
\begin{proof}
The theorem can be proved by application of the Cesari Theorem \cite[pg 68]{fleming1975deterministic}. Let the right hand side (RHS) of (\ref{eq:SIS_constraint}) be denoted by $f(u(t),i(t))$. The following requirements of the Theorem are met: $f(u(t),i(t))$ satisfies the required bound $|f(u(t),i(t))| \leq C_0(1+|i(t)|+|u(t)|)$ (with $C_0=\textnormal{sup}|\beta(t)+\gamma+u(t)|$; note that $\beta(t)$ is bounded and $0\leq u(t) \leq u_{max}$). The set $U$ and the set of solutions to initial value problem (\ref{eq:SIS_constraint}) and (\ref{eq:SIS_init_cond}) are non empty (due to Lipschitz continuity of $f(u(t),i(t))$ \cite[pg. 185]{birkhoff1989ordinary}). The control signal takes values in a closed set $[0, u_{max}]$. The cost due to the terminal state in the cost functional (\ref{eq:SIS_prob_reduced}) takes values in a compact interval $[0,1]$. The function $f(u(t),i(t))$ is linear in $u(t)$. In the cost functional (\ref{eq:SIS_prob_reduced}), the integrand, $bu^2(t)\geq C_1|u(t)|^{C_2}-C_3$. It is required that $C_1>0,~C_2>1$ which is satisfied if we choose $C_1=b,~C_2=1.5,~C_3=0$.
\end{proof}

\subsection{Solution to the SIS Optimal Control Problem}
\label{sec:SIS_pontryagin_numerical}
We use Pontryagin's Maximum principle \citep{kamien1991dynamic} to solve the optimal control problem (\ref{eq:SIS_prob_reduced}). Through this technique, we get a system of ordinary differential equations (ODEs) in terms of state and adjoint variables (with initial and boundary conditions, respectively) which are satisfied at the optimum. The system of ODEs can be solved numerically using boundary value ODE solvers. Let $\lambda(t)$ denote the adjoint variable. At time $t$, let $u^*(t)$ denote the optimum control and, $i^*(t)$ and $\lambda^*(t)$ the state and adjoint variables evaluated at the optimum.
\\ \emph{Hamiltonian:} The objective function in (\ref{eq:SIS_prob_reduced}) has been multiplied by $-1$ to convert the problem to a maximization problem.
\begin{align}
H(i(t),u(t),\lambda (t), t) = -b u^2(t) + \lambda (t)\Big[ -\beta(t) i^2(t) + (\beta(t) - \gamma - u(t))i(t) + u(t)\Big]. \nonumber
\end{align}
\emph{Adjoint equation:} $\dot{\lambda} ^*(t)$ is $- \frac{\partial}{\partial i(t)}H(i(t),u(t),\lambda (t), t)$ evaluated at the optimum.
\begin{eqnarray} \label{eqn:SIS-lambda(t)}
\dot{\lambda} ^*(t) &=& - \left.\frac{\partial}{\partial i(t)}H(i(t),u(t),\lambda (t), t) \right|_{\begin{smallmatrix}
i(t)=i^*(t), u(t)=u^*(t),\\
\lambda(t) = \lambda^*(t)
\end{smallmatrix}} \nonumber \\
&=& 2 \beta(t) i^*(t) \lambda ^*(t) - \big(\beta(t) - \gamma - u^*(t)\big) \lambda^*(t).
\label{eq:SIS_costate}
\end{eqnarray}
\emph{Hamiltonian Maximizing Condition:} At the interior points
\begin{eqnarray}
\left.\frac{\partial}{\partial u(t)} H(i(t),u(t),\lambda (t), t)\right|_{i(t)=i^*(t), u(t)=u^*(t), \lambda(t) = \lambda^*(t)} \nonumber \\
= -2 b u^*(t) - \lambda ^*(t)i^*(t) + \lambda^*(t) = 0. \nonumber
\end{eqnarray}
Hence the Hamiltonian maximizing condition leads to
\begin{equation}
 u^*(t) =
\begin{cases}
0 & \text{if } \frac{\lambda^*(t)(1-i^*(t))}{2b}<0,
\\
\frac{\lambda^*(t)(1-i^*(t))}{2b} & \text{if } 0\leq \frac{\lambda^*(t)(1-i^*(t))}{2b} \leq u_{max},
\\
u_{max} & \text{if } \frac{\lambda^*(t)(1-i^*(t))}{2b}>u_{max},
\end{cases}
~~~~~~\Rightarrow~~~~~~
u^*(t) = \text{min}\left\{\text{max}\left\{\frac{\lambda^*(t)(1-i^*(t))}{2b},0\right\},u_{max}\right\}.
\label{eq:SIS_u}
\end{equation}
\emph{Transversality condition:} From the transversality condition we get
\begin{eqnarray}
\lambda ^*(T) = 1. \label{eq:SIS_transversality_cond}
\end{eqnarray}

Substituting (\ref{eq:SIS_u}) in (\ref{eq:SIS_constraint}) and (\ref{eq:SIS_costate}) \big(and using the initial condition (\ref{eq:SIS_init_cond}) and boundary condition (\ref{eq:SIS_transversality_cond})\big), we get a system of ODEs which can be solved using standard boundary value problem ODE solving techniques. We have implemented the shooting method \cite{kutz2005practical} in MATLAB to solve the boundary value problem. Equations (\ref{eq:SIS_constraint}) and (\ref{eq:SIS_costate}) are solved using MATLAB's initial value problem solver \texttt{ode45()} with Equation (\ref{eq:SIS_costate}) initialized arbitrarily. Naturally, the solution will not satisfy the required boundary condition (\ref{eq:SIS_transversality_cond}); hence the estimation of the initial condition of Equation (\ref{eq:SIS_costate}) is improved using the optimization routine \texttt{fminunc()} until the boundary condition (\ref{eq:SIS_transversality_cond}) is met with desired accuracy. Another option is to use forward-backward sweep method explained in \cite{asano2008optimal,gaff2009optimal}.

\subsection{Uniqueness of the Solution to the SIS Optimal Control Problem}
\begin{thm}
\label{theorem:SIS_uniqueness}
For a sufficiently small campaign deadline, $T$, the state and adjoint trajectories at the optimum and the optimal control to problem (\ref{eq:SIS_prob_reduced}) are unique.
\end{thm}
\begin{proof}
The proof technique is same as in \cite{fister1998optimizing}, details are in \ref{app:uniqueness_SIS}.
\end{proof}

\section{Analysis of the Controlled SIR Epidemic}
\label{sec:SIR_model_analysis}

After removing the redundant information, problem (\ref{eq:prob_SIR_word_of_mouth}) can be rewritten in terms of two state variables and two controls as follows (sets $U_1$ and $U_2$ are according to Definition \ref{def:SIR_U1_U2}):
\begin{eqnarray}
\underset{u_1\in U_1,u_2\in U_2}{\text{min }} J &=& -1 + s(T) + \int_{0}^{T} \hspace{-1em}\big(bu_1^2(t)+cu_2^2(t)\big)dt \label{eq:SIR_wom_prob_modified} \\
\text{subject to:  } \dot{s}(t) &=& -\big(\beta(t)+u_2(t)\big) s(t) \big(1-s(t)-r(t)\big) - u_1(t)s(t) \label{eq:SIR_modified_constraint1} \\
\dot{r}(t) &=& \gamma \big(1-s(t)-r(t)\big)  \label{eq:SIR_modified_constraint2} \\
0& \leq & s(t),r(t)\ \leq \ 1   \nonumber \\
s(0) & = & 1 - i_0, \ \ r(0) = 0. \label{eq:SIR_wom_init_cond}
\end{eqnarray}

\subsection{Existence of a Solution}

\begin{thm}
There exist an optimal control signals $u_1\in U_1,u_2\in U_2$ and corresponding solutions $s^*(t), r^*(t)$ to the initial value problem (\ref{eq:SIR_modified_constraint1}), (\ref{eq:SIR_modified_constraint2}) and (\ref{eq:SIR_wom_init_cond}) such that $(u_1,u_2)^T\in\underset{u_1\in U_1, u_2\in U_2}{\textnormal{argmin }}\{J(u_1,u_2)\}$ in problem (\ref{eq:SIR_wom_prob_modified}).
\end{thm}
\begin{proof}
The theorem can be proved by application of the Cesari Theorem \cite[pg 68]{fleming1975deterministic}. The details are omitted.
\end{proof}

\subsection{Solution to the SIR Optimal Control Problem with Direct and Word-of-mouth Control}
Let $\lambda_s(t)$ and $\lambda_r(t)$ be the adjoint variables. At time $t$, let $u_1^*(t),~u_2^*(t)$ denote the optimum controls and, $s^*(t),~r^*(t)$ and $\lambda_s^*(t),~\lambda_r^*(t)$ the state and adjoint variables evaluated at the optimum. Using Pontryagin's Maximum Principle \citep{kamien1991dynamic} we get the following equations.
\\ \emph{Hamiltonian:} The objective function in (\ref{eq:SIR_wom_prob_modified}) has been multiplied by $-1$ to convert the problem to maximization problem.
\begin{align}
& H(s(t),r(t),u_1(t),u_2(t),\lambda_s (t),\lambda_r (t), t) \nonumber\\
=& -b u_1^2(t) -c u_2^2(t) + \lambda_s (t)\Big[ -\big(\beta(t)+u_2(t)\big) s(t) + \big(\beta(t)+u_2(t)\big) s^2(t) + \big(\beta(t)+u_2(t)\big) s(t)r(t) - u_1(t)s(t) \Big] \nonumber \\
&  + \lambda_r(t) \Big[\gamma - \gamma s(t) - \gamma r(t)\Big]. \nonumber
\end{align}
\emph{Adjoint Equations:}
\begin{eqnarray}
\hspace{-1em}\dot{\lambda}_s ^{*}(t) 
&=& \beta(t) \lambda_s ^{*}(t) - 2 \beta(t) \lambda_s ^{*}(t)s^{*}(t) - \beta(t) \lambda_s ^{*}(t)r^*(t) + \lambda_s ^{*}(t)u_1^*(t) + \lambda_s^*(t) u_2^*(t) - 2\lambda_s^*(t) u_2^*(t) s^*(t) \nonumber \\
& & - \lambda_s^*(t) u_2^*(t) r^*(t) + \gamma \lambda_r ^{*}(t) \label{eq:SIR_wom_lam_s} \\
\hspace{-1em}\dot{\lambda}_r ^{*}(t) 
&=& -\beta(t) \lambda_s ^{*}(t)s^*(t) + \lambda_s^*(t) u_2^*(t) s^*(t) + \gamma \lambda_r ^{*}(t) \label{eq:SIR_wom_lam_r}
\end{eqnarray}
\emph{Hamiltonian Maximizing Condition:} Derivative of the Hamiltonian evaluates to zero at interior points, 
hence the Hamiltonian maximizing condition leads to
\begin{equation}
u_1^*(t) = \text{min}\left\{\text{max}\left\{\frac{\lambda_s^*(t)s^*(t)}{-2b},0\right\},u_{1max}\right\}.
\label{eq:SIR_wom_u1}
\end{equation}
and,
\begin{equation}
u_2^*(t) = \text{min}\left\{\text{max}\left\{\frac{\lambda_s^*(t)s^*(t)\big(1-2s^*(t)-r^*(t)\big)}{-2c},0\right\},u_{2max}\right\}.
\label{eq:SIR_wom_u2}
\end{equation}
\emph{Transversality condition:} $\lambda_s ^*(T) = -1$ and $\lambda_r ^*(T) = 0$.

Substituting the values of $u_1^*(t)$ and $u_2^*(t)$ from (\ref{eq:SIR_wom_u1}) and (\ref{eq:SIR_wom_u2}) to (\ref{eq:SIR_modified_constraint1}), (\ref{eq:SIR_modified_constraint2}), (\ref{eq:SIR_wom_lam_s}) and (\ref{eq:SIR_wom_lam_r}) and solving the system of ODEs numerically using the technique described in Section \ref{sec:SIS_pontryagin_numerical}, we can compute the state and adjoint variables and hence the optimal control.

\subsection{Uniqueness of the Solution to the SIR Optimal Control Problem with Direct Recruitment and Word-of-mouth Control}
\begin{thm}
\label{theorem:SIR_wom_uniqueness}
For a sufficiently small campaign deadline, $T$, the state and adjoint trajectories at the optimum and the solution to the optimal control problem (\ref{eq:SIR_wom_prob_modified}) are unique.
\end{thm}
\begin{proof}
The proof technique is same as in \cite{fister1998optimizing}, details are in \ref{app:uniqueness_SIR_wom}.
\end{proof}

\section{Results}
\label{sec:Results}

We divide this section into three parts. Section \ref{sec:results-constant_spreading_rate} studies the control signal and corresponding state evolution for constant spreading rate and Section \ref{sec:results-variable_spreading_rate} for variable spreading rate. The tree in Fig. \ref{fig:sec6_tree} shows how the results are organized in these two subsections. In Section \ref{sec:results-J_vs_parameters} we study the role played by various parameters ($\beta, \gamma, T, b, c$) on the cost functional $J$ for both SIS and SIR models.

\begin{figure}[t!]
\centering
\includegraphics[width=100mm]{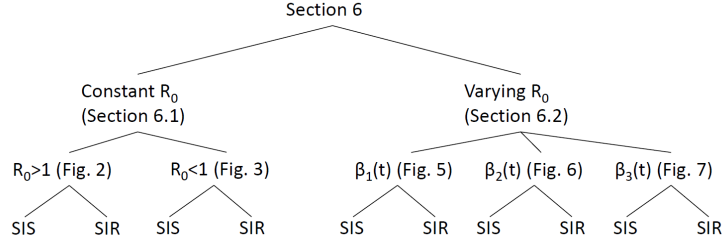}
\caption{Organization of plots related to the shapes of the control signals and state evolutions in Sections \ref{sec:results-constant_spreading_rate} and \ref{sec:results-variable_spreading_rate}. In addition, plots in Section \ref{sec:results-J_vs_parameters} show variation of the cost functional $J$ with respect to various parameters in the SIS and SIR models.}
\label{fig:sec6_tree}
\end{figure}

\begin{figure}[ht!]
\subfloat[SIS epidemic, $\beta=1$ ($R_0=10$) and $\beta=2$ ($R_0=20$) Parameter values: $\gamma=0.1, T=5, b=15, u_{max}=0.06, i_0=0.01$. \label{fig:control_state_sis_constantR0_beta1}]{
\includegraphics[width=65mm]{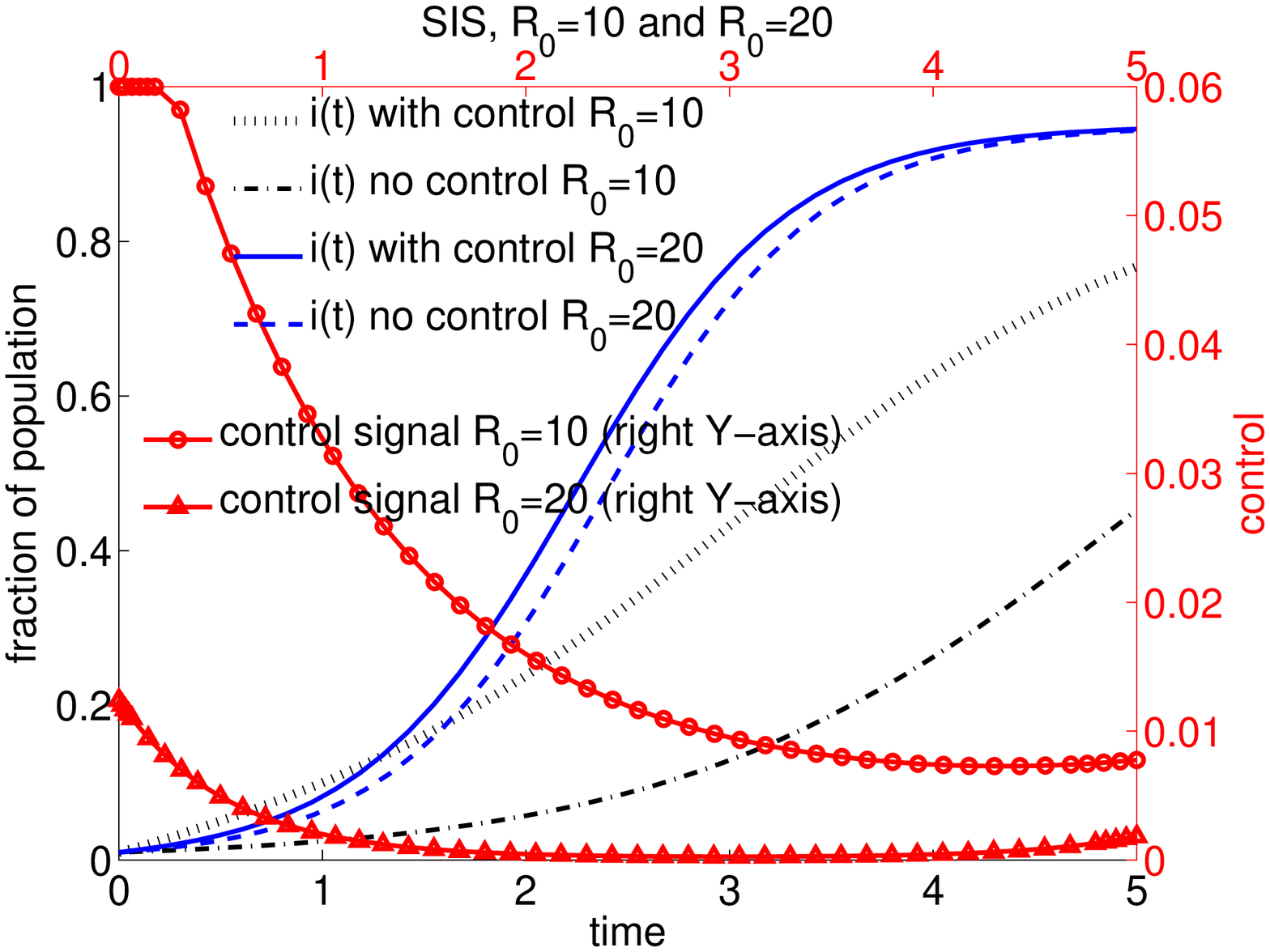} }
\hfill
\subfloat[SIR epidemic, $\beta=1$ or $R_0=10$ Parameter values: $\gamma=0.1, T=5, b=15, c=1, u_{1max}=0.06, u_{2max}=0.3, s_0=0.99, i_0=0.01$. \label{fig:control_state_sir_wom_constantR0_beta1}]{
\includegraphics[width=65mm]{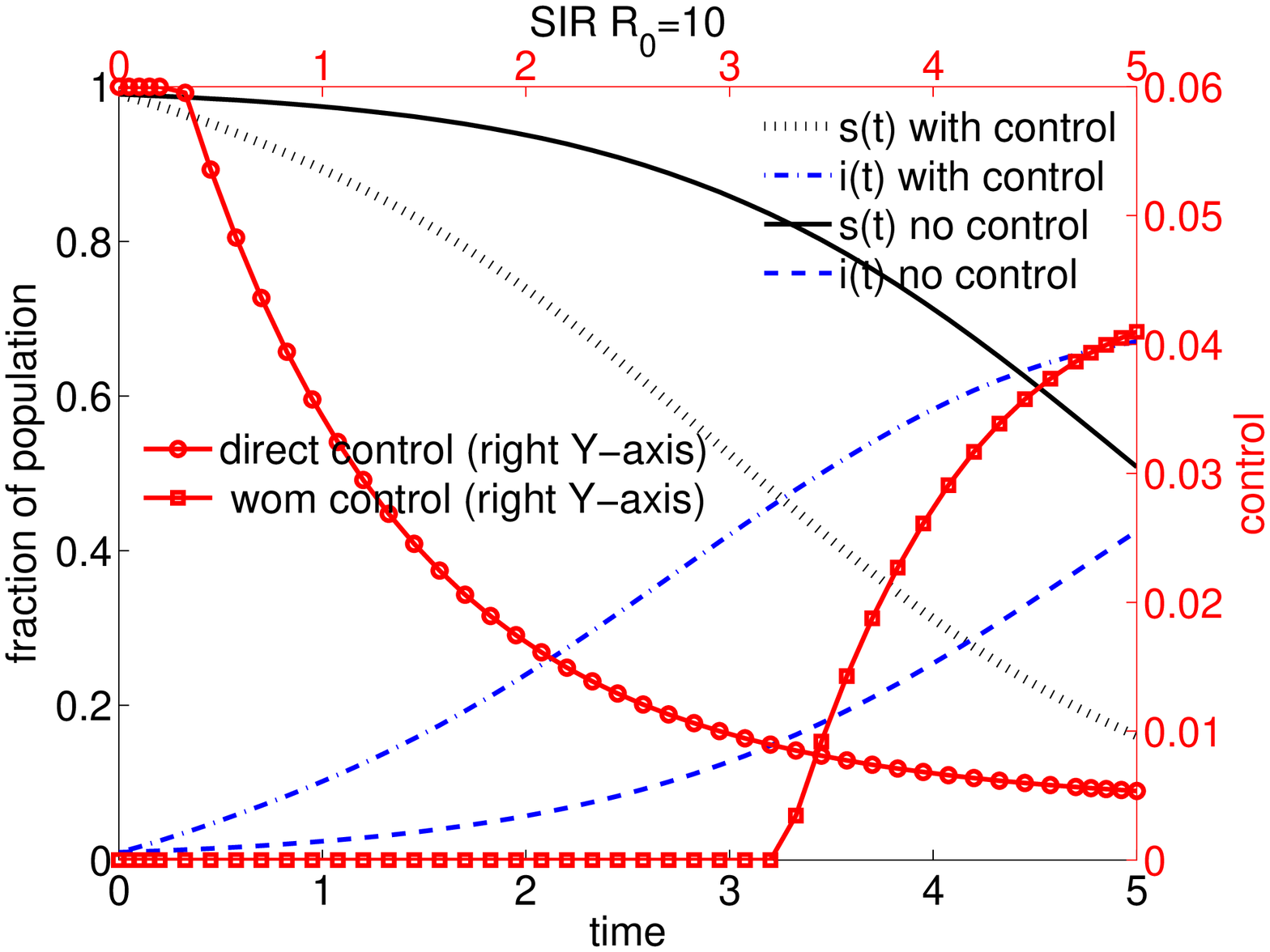} }
\caption{Optimal control, state evolutions with control and state evolutions without control for the SIS epidemic (model in Section \ref{sec:SIS_sys_model}) and the SIR epidemic (model in Section \ref{sec:SIR_sys_model}). \emph{Note that state variables are plotted with respect to the left Y-axis and control signals are plotted with respect to the right Y-axis.}}
\label{fig:control_state_sis_constantR0}
\end{figure}

\begin{figure}[ht!]
\subfloat[SIS epidemic, $\beta=.03$ or $R_0=0.3$ Parameter values: $\gamma=0.1, T=5, b=15, u_{max}=0.06, i_0=0.01$ \label{fig:control_state_sis_constantR0_beta03}]{
\includegraphics[width=65mm]{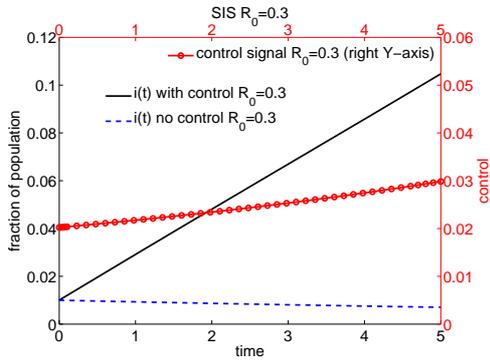} }
\hfill
\subfloat[SIR epidemic, $\beta=.03$ or $R_0=0.3$ Parameter values: $\gamma=0.1, T=5, b=15, c=1, u_{1max}=0.06, u_{3max}=0.3, s_0=0.99, i_0=.01$ \label{fig:control_state_sir_wom_constantR0_beta03}]{
\includegraphics[width=65mm]{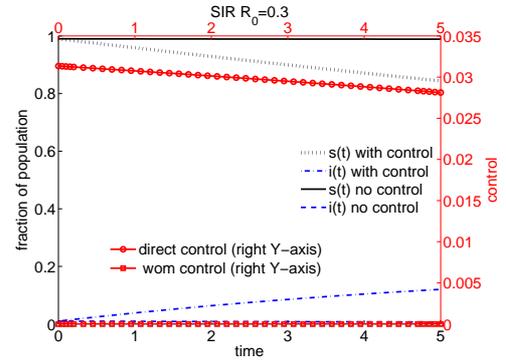} }
\caption{Optimal control, state evolution with control and state evolution without control for the SIS epidemic (model in Section \ref{sec:SIS_sys_model}) and the SIR epidemic (model in Section \ref{sec:SIR_sys_model}). \emph{Note that state variables are plotted with respect to the left Y-axis and control signals are plotted with respect to the right Y-axis.}}
\label{fig:control_state_sir_wom_constantR0}
\end{figure}

\subsection{Constant Effective Spreading Rate Over Time}
\label{sec:results-constant_spreading_rate}

First we consider the case when the effective spreading rate is constant over time. Thus, $\beta(t)=\beta,~\forall t\in[0,T].$ The basic reproductive number, $R_0$, for an epidemic is defined as the expected number of secondary infections caused by an infected node in the early stages of epidemic outbreak. For the (uncontrolled) SIS and SIR epidemic considered in this paper, $R_0=\beta/\gamma$, the ratio of effective spreading rate to the recovery rate \citep{barrat2008dynamical}. Give a campaign deadline $T$, basic reproductive number $R_0$ captures how viral the information epidemic is. Qualitatively, increasing $\beta$ or decreasing $\gamma$, while holding the other constant, is expected to have same result. If $R_0>1$ for the uncontrolled system, the epidemic, on the average will become endemic and for $R_0<1$, the epidemic dies out with probability $1$ \cite{barrat2008dynamical}. We discuss these cases separately:

\begin{enumerate}[(i)]
\item With reference to Figs. \ref{fig:control_state_sis_constantR0_beta1} (for SIS) and \ref{fig:control_state_sir_wom_constantR0_beta1} (for SIR), where $R_0>1$, direct control signal is strong when the targeted population (susceptibles) are in abundance (at beginning of the campaign period) and vice versa. Early infection increases the extent of information spreading as the system has a tendency to sustain the population in the infected state (because infection is faster than recovery). Also, the word-of-mouth control in Fig. \ref{fig:control_state_sir_wom_constantR0_beta1} switches on from zero when populations of both susceptible and infected individuals reach significant levels. Providing word-of-mouth incentive is effective only when there are substantial number of infected nodes as well as enough number of susceptibles to convince.

\item When $R_0<1$ for the uncontrolled system, the uncontrolled information epidemic dies out. The control and state evolutions for this case are shown in Figs. \ref{fig:control_state_sis_constantR0_beta03} (for SIS) and \ref{fig:control_state_sir_wom_constantR0_beta03} (for SIR).  We find the direct control signal to be less variable over time. In fact, in the SIS case, it increases with time. Fast recovery (compared to infection) of the nodes makes a strong control at the beginning stages of the epidemic ineffective. Also, notice that in the SIR case (Fig. \ref{fig:control_state_sir_wom_constantR0_beta03}), the optimal strategy advocates not using word-of-mouth control throughout the campaign duration.

\item We make an additional observation with reference to the control signals and the state evolution curves plotted for $R_0=10$ and $R_0=20$ in Fig. \ref{fig:control_state_sis_constantR0_beta1}. Given a campaign deadline $T$, as $R_0=\beta/\gamma$ increases, (a) the control effort (measured by area under the control curve) decreases and (b) control signal has limited effect on system evolution. Thus, campaigns which are less viral will benefit more from the application of optimal control than the campaigns which are more viral. Such an observation has implications on marketing strategies for new products launched by a reputed company compared to a newbie in the market, or publicity of a movie by a famous director compared to a newcomer. Similar observations were made for direct and word-of-mouth controls in case of the SIR model, but the curves corresponding to $R_0=20$ are omitted for brevity.
\end{enumerate}

\begin{figure}[ht!]
\centering
\includegraphics[width=65mm]{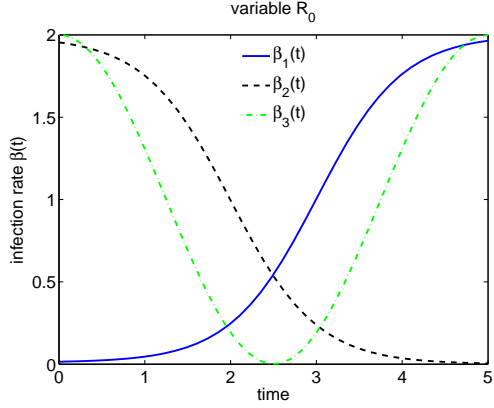}
\caption{Time varying effective spreading rate, $\beta_1(t), \beta_2(t)$ and $\beta_3(t)$ defined in equations (\ref{eq:beta1t}), (\ref{eq:beta2t}) and (\ref{eq:beta3t}) respectively. Parameter values: $\beta_m=.01, \beta_M=2, T=5, a_1 = 2, c_1 = 3; a_2=2, c_2=2, c_m=1, c_a=1$ and $t\in[0,5]$.}
\label{fig:beta_sis_sir_variableR0}
\end{figure}

\begin{figure}[ht!]
\subfloat[SIS epidemic. Parameter values: $\gamma=0.1, T=5, b=15, u_{max}=0.06, i_0=0.01$. \label{fig:control_state_sis_variableR0_beta1}]{
\includegraphics[width=65mm]{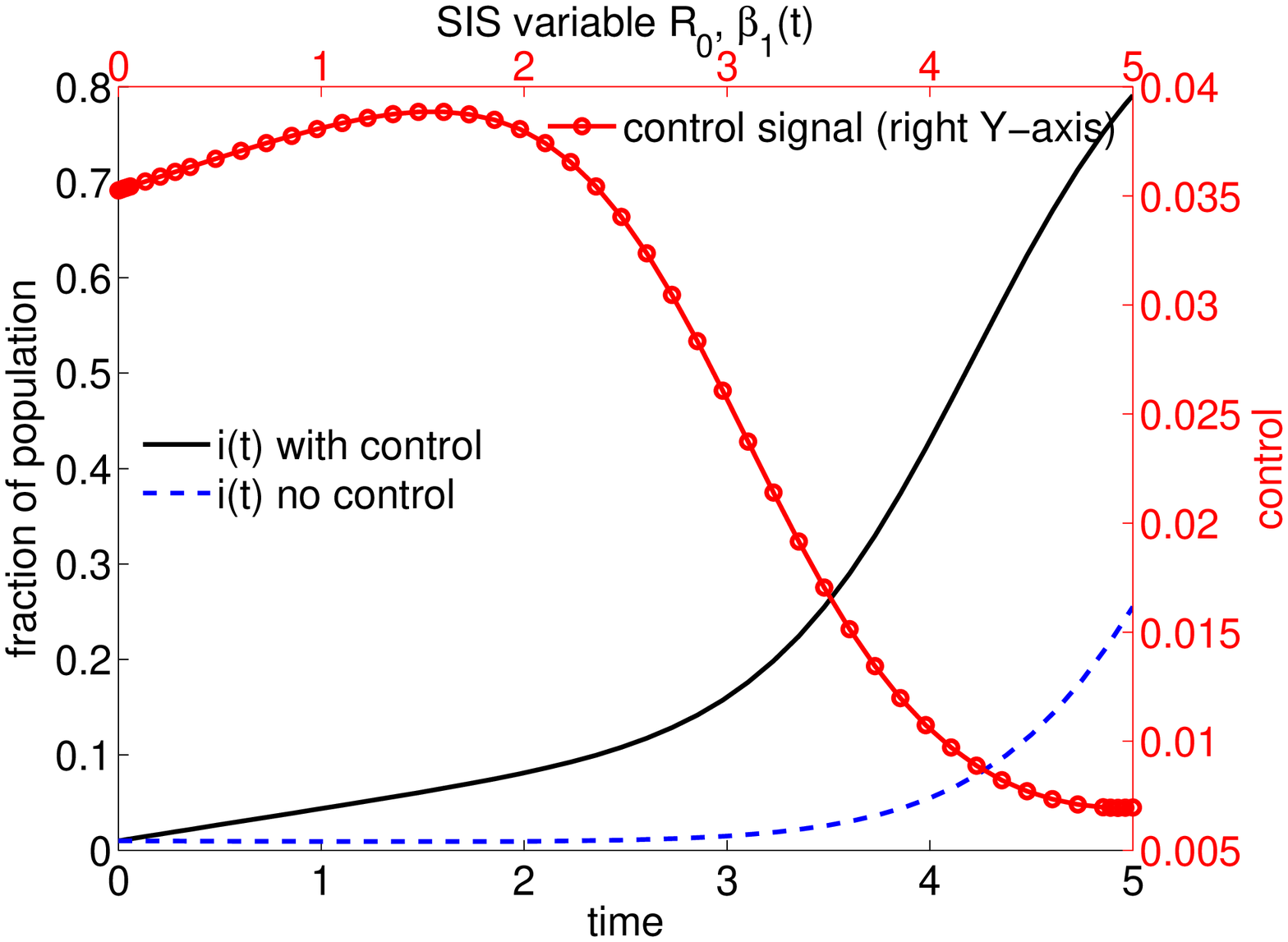} }
\hfill
\subfloat[SIR epidemic. Parameter values: $\gamma=0.1, T=5, b=15, c=1, u_{1max}=0.06, u_{2max}=0.3, s_0=0.99, i_0=0.01$. \label{fig:control_state_sir_variableR0_beta1}]{
\includegraphics[width=65mm]{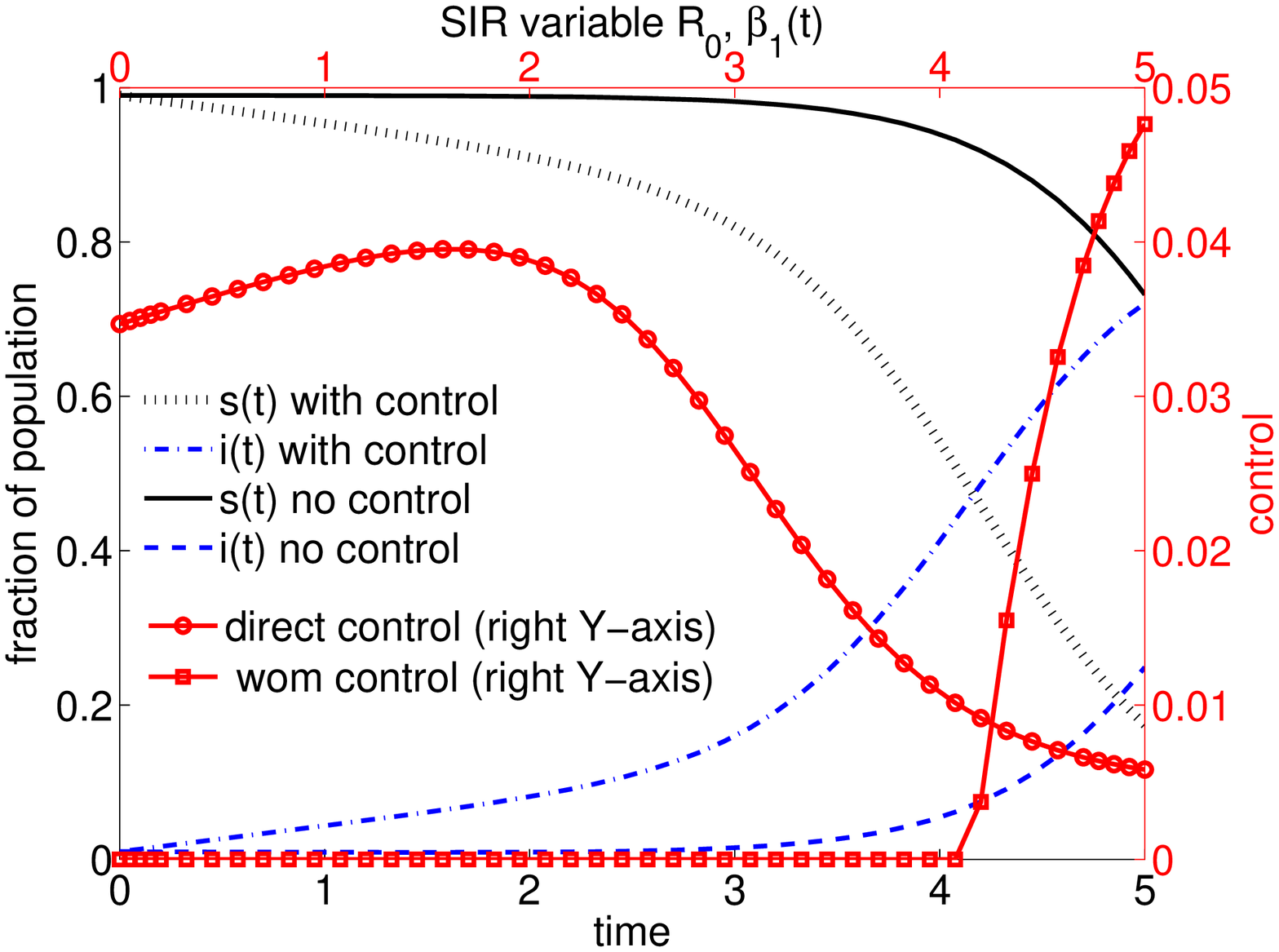} }
\caption{Optimal control, state evolution with control and state evolution without control for the SIS and SIR epidemic, for time-varying spreading rate $\beta_1(t)$. \emph{Note that state variables are plotted with respect to the left Y-axis and control signals are plotted with respect to the right Y-axis.}}
\label{fig:control_state_variableR0_beta1}
\end{figure}

\begin{figure}[ht!]
\subfloat[SIS epidemic. Parameter values: $\gamma=0.1, T=5, b=15, u_{max}=0.06, i_0=0.01$. \label{fig:control_state_sis_variableR0_beta2}]{
\includegraphics[width=70mm]{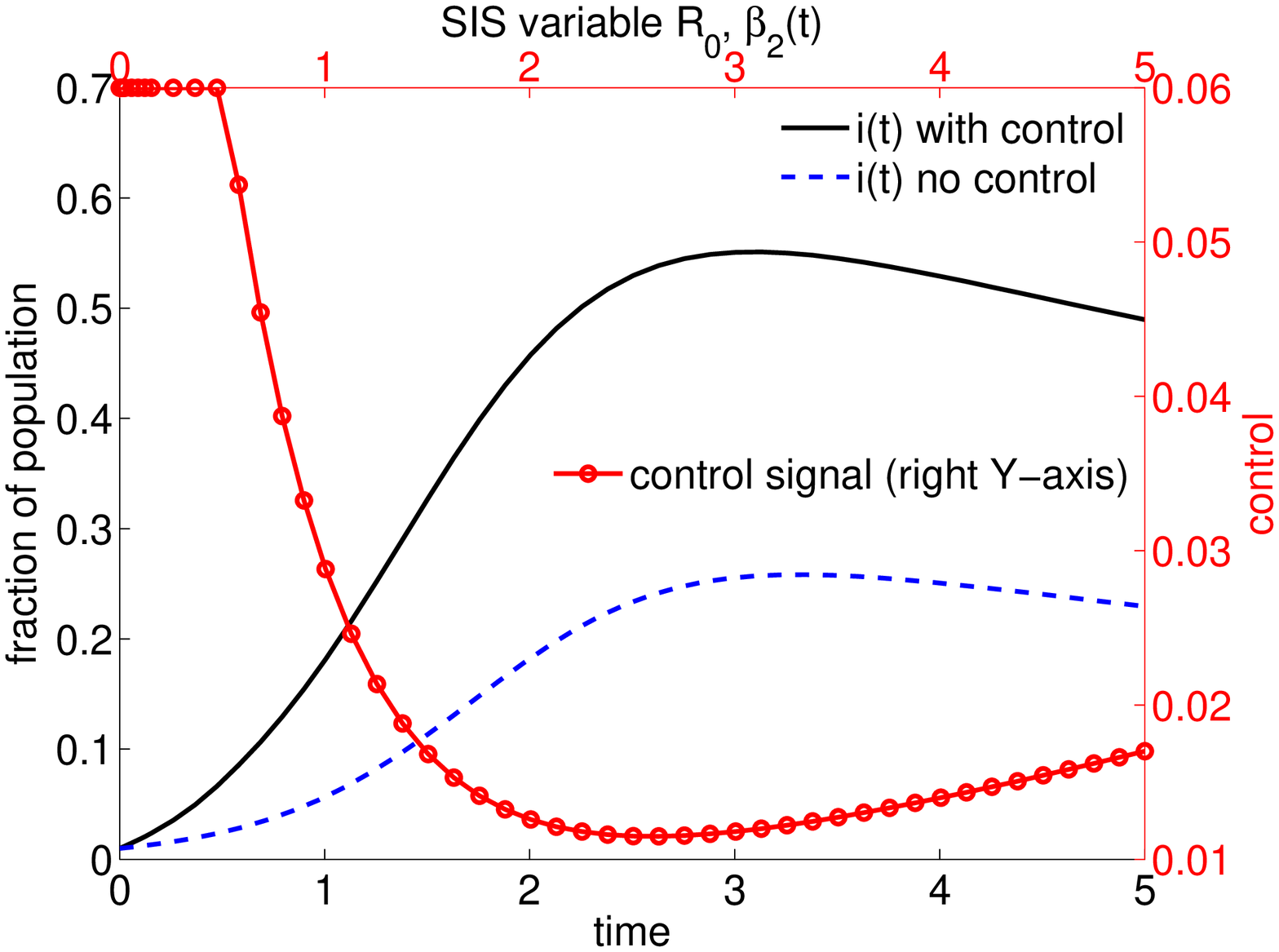} }
\hfill
\subfloat[SIR epidemic. Parameter values: $\gamma=0.1, T=5, b=15, c=1, u_{1max}=0.06, u_{2max}=0.3, s_0=0.99, i_0=0.01$. \label{fig:control_state_sir_variableR0_beta2}]{
\includegraphics[width=70mm]{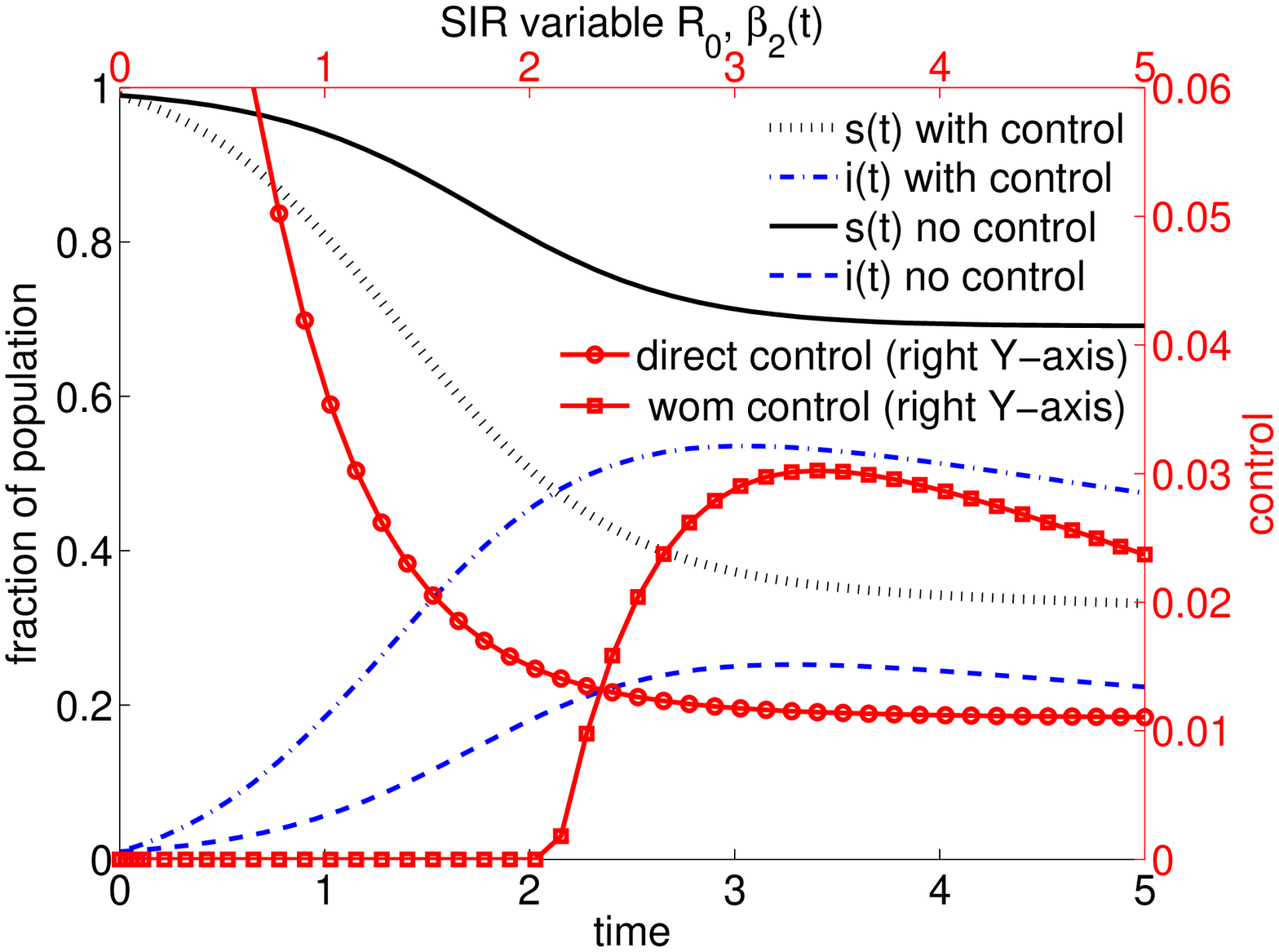} }
\caption{Optimal control, state evolution with control and state evolution without control for the SIS and SIR epidemics for time-varying spreading rate $\beta_2(t)$. \emph{Note that state variables are plotted with respect to the left Y-axis and control signals are plotted with respect to the right Y-axis.}}
\label{fig:control_state_variableR0_beta2}
\end{figure}

\begin{figure}[ht!]
\subfloat[SIS epidemic. Parameter values: $\gamma=0.1, T=5, b=15, u_{max}=0.06, i_0=0.01$. \label{fig:control_state_sis_variableR0_beta3}]{
\includegraphics[width=70mm]{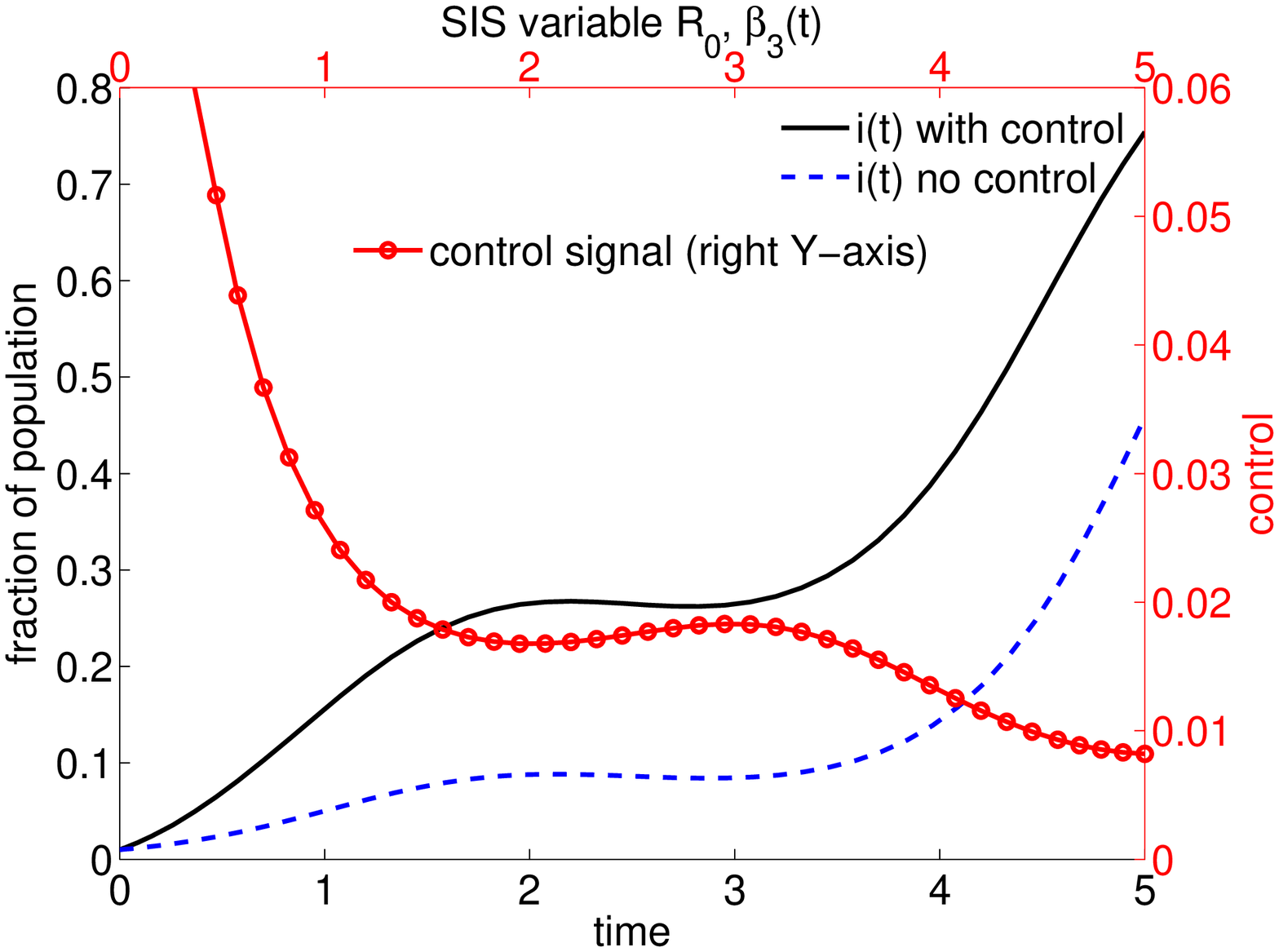} }
\hfill
\subfloat[SIR epidemic. Parameter values: $\gamma=0.1, T=5, b=15, c=1, u_{1max}=0.06, u_{2max}=0.3, s_0=0.99, i_0=0.01$. \label{fig:control_state_sir_variableR0_beta3}]{
\includegraphics[width=70mm]{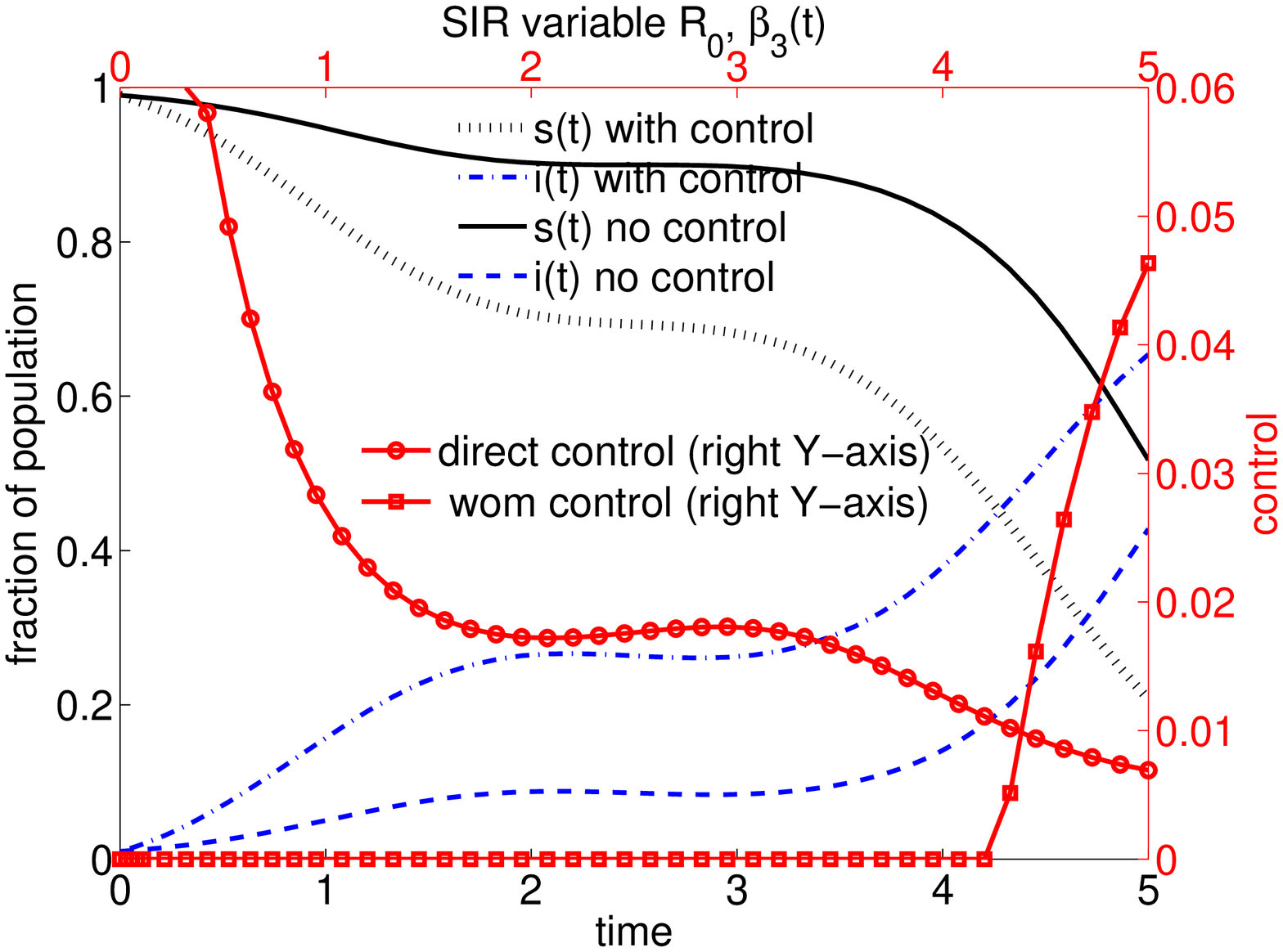} }
\caption{Optimal control, state evolution with control and state evolution without control for the SIS and SIR epidemics for time-varying spreading rate $\beta_3(t)$. \emph{Note that state variables are plotted with respect to the left Y-axis and control signals are plotted with respect to the right Y-axis.}}
\label{fig:control_state_variableR0_beta3}
\end{figure}

\subsection{Variable Effective Spreading Rate Over Time}
\label{sec:results-variable_spreading_rate}

To model the varying interest of a population in spreading the information during the campaign period, we consider three different functions $\beta_1(t), \beta_2(t)$ and $\beta_3(t)$. We model the cases of increasing, decreasing and fluctuating interests as we approach the deadline through these functions. The functions are increasing sigmoid, decreasing sigmoid and cosine (plotted in Fig. \ref{fig:beta_sis_sir_variableR0}) and are defined as:
\begin{eqnarray}
\beta_1(t) & = & \beta_m + \left( \frac{\beta_M-\beta_m}{1+e^{-a_1(t-c_1)}} \right), \label{eq:beta1t} \\
\beta_2(t) & = & (\beta_M - \beta_m) \left(1 - \frac{1}{1+e^{-a_2(t-c_2)}} \right), \label{eq:beta2t} \\
\beta_3(t) & = & c_m + c_a\cos(2\pi t/T), \label{eq:beta3t}
\end{eqnarray}
where the values of the parameters used are: $\beta_m=.01, \beta_M=2, T=5, a_1=2, c_1=3, a_2=2, c_2=2, c_m=1, c_a=1$ and $t\in[0,5]$. Wherever $\beta_i(t),~i=1,2,3$ are used, the recovery rate is set to $\gamma=0.1$. The increasing effective spreading rate, $\beta_1(t)$ may represent the increasing interest of people to talk about election candidates as we approach the polling date. The decreasing effective spreading rate, $\beta_2(t)$ may represent gradual loss of interest of people in talking about some newly launched product (\emph{e.g.} a computer game) after its release. Fluctuating effective spreading rate $\beta_3(t)$ may represent changes in demand of a product/service with time (\emph{e.g.}, movie tickets for weekend shows may have higher demand than tickets for weekday shows). Depending on the application, other profiles for $\beta(t)$ are possible.

The controls and state evolutions for the SIS and SIR models discussed in this paper for the time-varying effective spreading rates defined in Equations (\ref{eq:beta1t}), (\ref{eq:beta2t}) and (\ref{eq:beta3t}) are plotted in Figs. \ref{fig:control_state_variableR0_beta1}--\ref{fig:control_state_variableR0_beta3}. The shape of the optimal control may be different from the case when the effective spreading rate $\beta$ is a constant over time. This shows the need to determine the interest level of the population (and hence $\beta(t)$) before deciding on the optimal control strategy. The figures also show the effectiveness of the optimal control strategy over the case when no control is used, in increasing the number of infected nodes in the SIS model, and number of infected and recovered nodes in the SIR model. Thus, optimal campaigning is beneficial in real world scenarios where the effective information spreading rate may be variable.

The recovery rate may also be a time dependent quantity in real world applications. The framework developed in this paper can be easily modified to include both time dependent $\beta(t)$ and $\gamma(t)$.

\begin{figure}[ht!]
\subfloat[Different control signals for SIS model described in Section \ref{sec:SIS_sys_model}. \label{fig:how_different_controls_and_states_look_control}]{
\includegraphics[width=65mm]{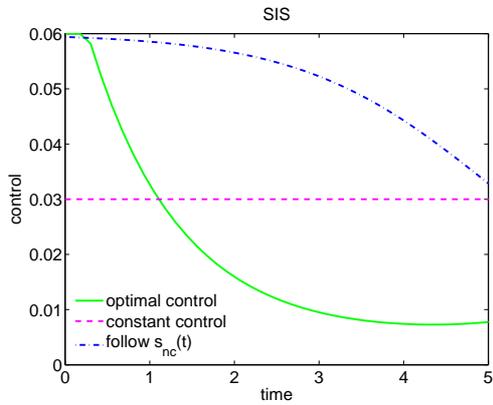} }
\hfill
\subfloat[Evolution of the proportion of infected individuals for different control strategies. \label{fig:how_different_controls_and_states_look_states}]{
\includegraphics[width=65mm]{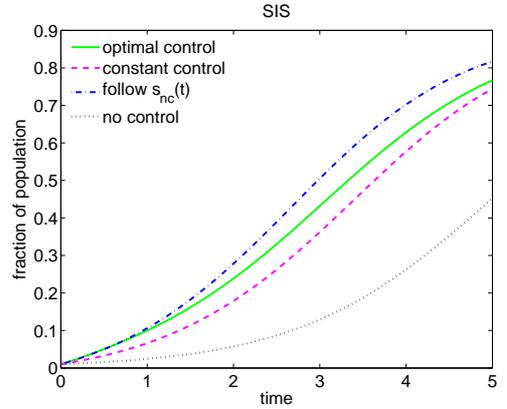} }
\caption{Shapes of different control signals and the corresponding state evolutions for the SIS model. Parameter values: $\beta = 1,~ \gamma=0.1,~ (R_0=10),~ T=5,~ b=15,~ u_{max}=0.06,~ i_0=0.1$. Note: both of the Figures contribute to the objective function which is optimized, see Figs. \ref{fig:J_vs_beta_sis}--\ref{fig:J_vs_b_sis} for objective function values for different strategies with varying parameter values.}
\label{fig:how_different_controls_and_states_look}
\end{figure}

\begin{figure}[ht!]
\subfloat[SIS epidemic. Parameter values: $\gamma=0.1, T=5, b=15, i_0=0.01, u_{max}=0.06$. \label{fig:J_vs_beta_sis}]{
\includegraphics[width=65mm]{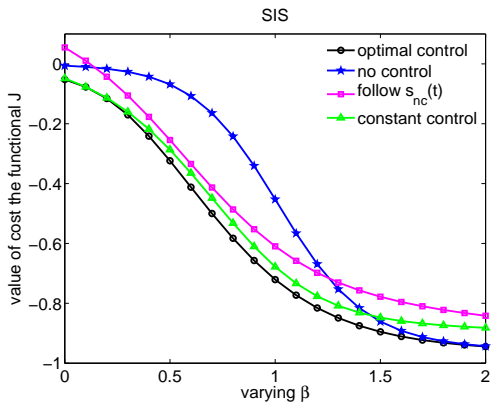} }
\hfill
\subfloat[SIR epidemic. Parameter values: $\gamma=0.1, T=5, b=15, c=1, u_{1max}=0.06, u_{2max}=0.3, s_0=0.99, i_0=0.01$. \label{fig:J_vs_beta_sir}]{
\includegraphics[width=65mm]{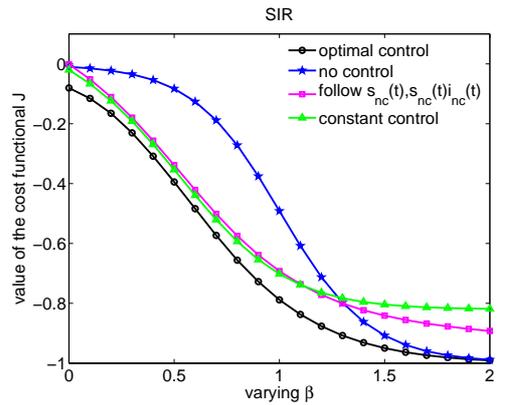} }
\caption{Objective functional $J$ vs spreading rate $\beta$ for different control strategies.}
\label{fig:J_vs_beta}
\end{figure}

\begin{figure}[ht!]
\subfloat[SIS epidemic. Parameter values: $\beta=1, T=5, b=15, i_0=0.01, u_{max}=0.06$. \label{fig:J_vs_gamma_sis}]{
\includegraphics[width=65mm]{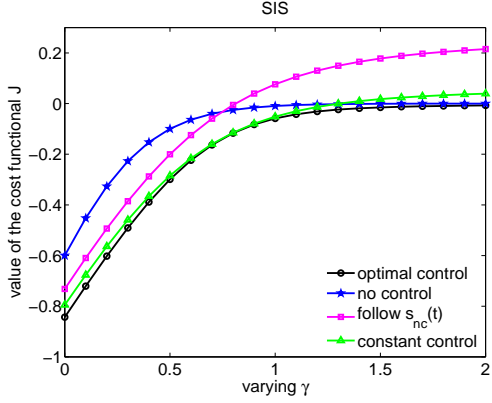} }
\hfill
\subfloat[SIR epidemic. Parameter values: $\beta=1, T=5, b=15, c=1, u_{1max}=0.06, u_{2max}=0.3, s_0=0.99, i_0=0.01$. \label{fig:J_vs_gamma_sir}]{
\includegraphics[width=65mm]{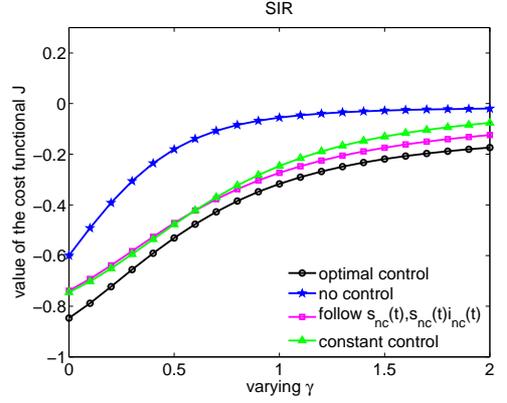} }
\caption{Objective functional $J$ vs recovery rate $\gamma$ for different control strategies.}
\label{fig:J_vs_gamma}
\end{figure}

\begin{figure}[ht!]
\subfloat[SIS epidemic. Parameter values: $\beta=1, \gamma=0.1, b=15, i_0=0.01, u_{max}=0.06$. \label{fig:J_vs_Tdeadline_sis}]{
\includegraphics[width=65mm]{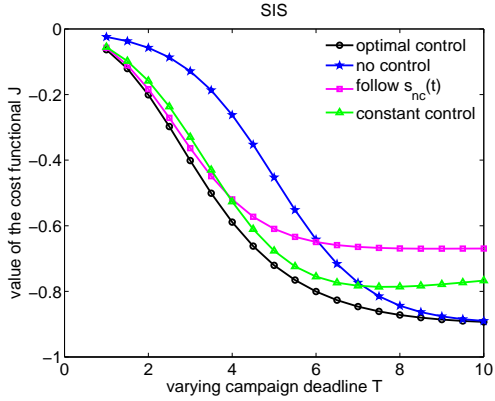} }
\hfill
\subfloat[SIR epidemic. Parameter values: $\beta=1, \gamma=0.1, b=15, c=1, u_{1max}=0.06, u_{2max}=0.3, s_0=0.99, i_0=0.01$. \label{fig:J_vs_Tdeadline_sir}]{
\includegraphics[width=65mm]{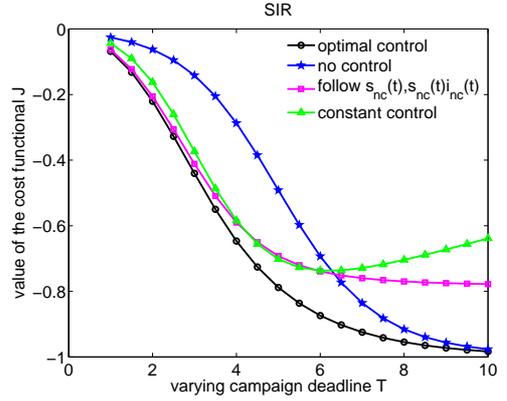} }
\caption{Objective functional $J$ vs campaign deadline $T$ for different control strategies.}
\label{fig:J_vs_Tdeadline}
\end{figure}

\begin{figure}[ht!]
\subfloat[SIS epidemic. Parameter values: $\beta=1, \gamma=0.1, T=5, i_0=0.01, u_{max}=0.06$. \label{fig:J_vs_b_sis}]{
\includegraphics[width=65mm]{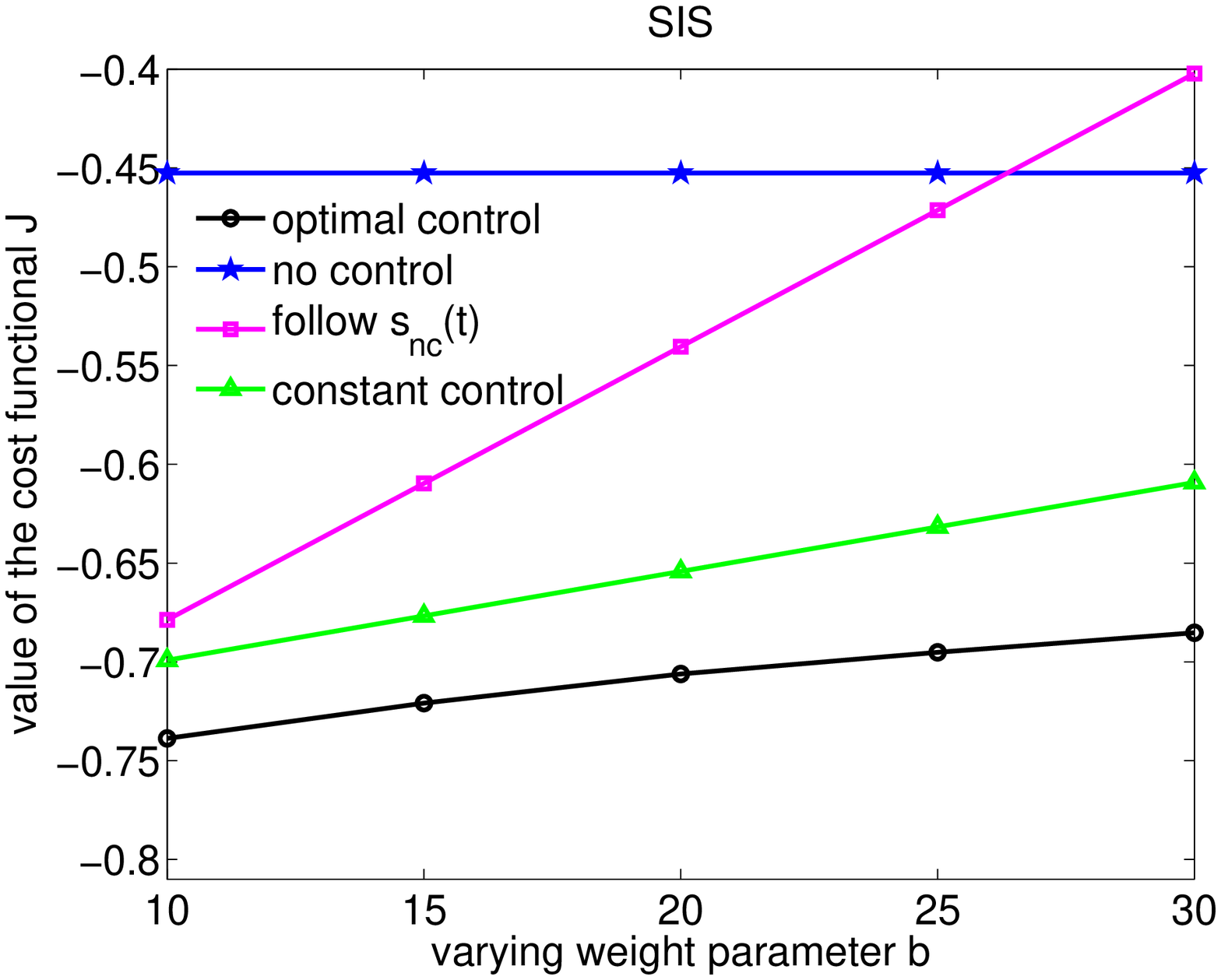} }
\hfill
\subfloat[SIR epidemic. Parameter values: $\beta=1, \gamma=0.1, T=5, u_{1max}=0.06, u_{2max}=0.3, s_0=0.99, i_0=0.01$. \label{fig:J_vs_b_c_sir}]{
\includegraphics[width=65mm]{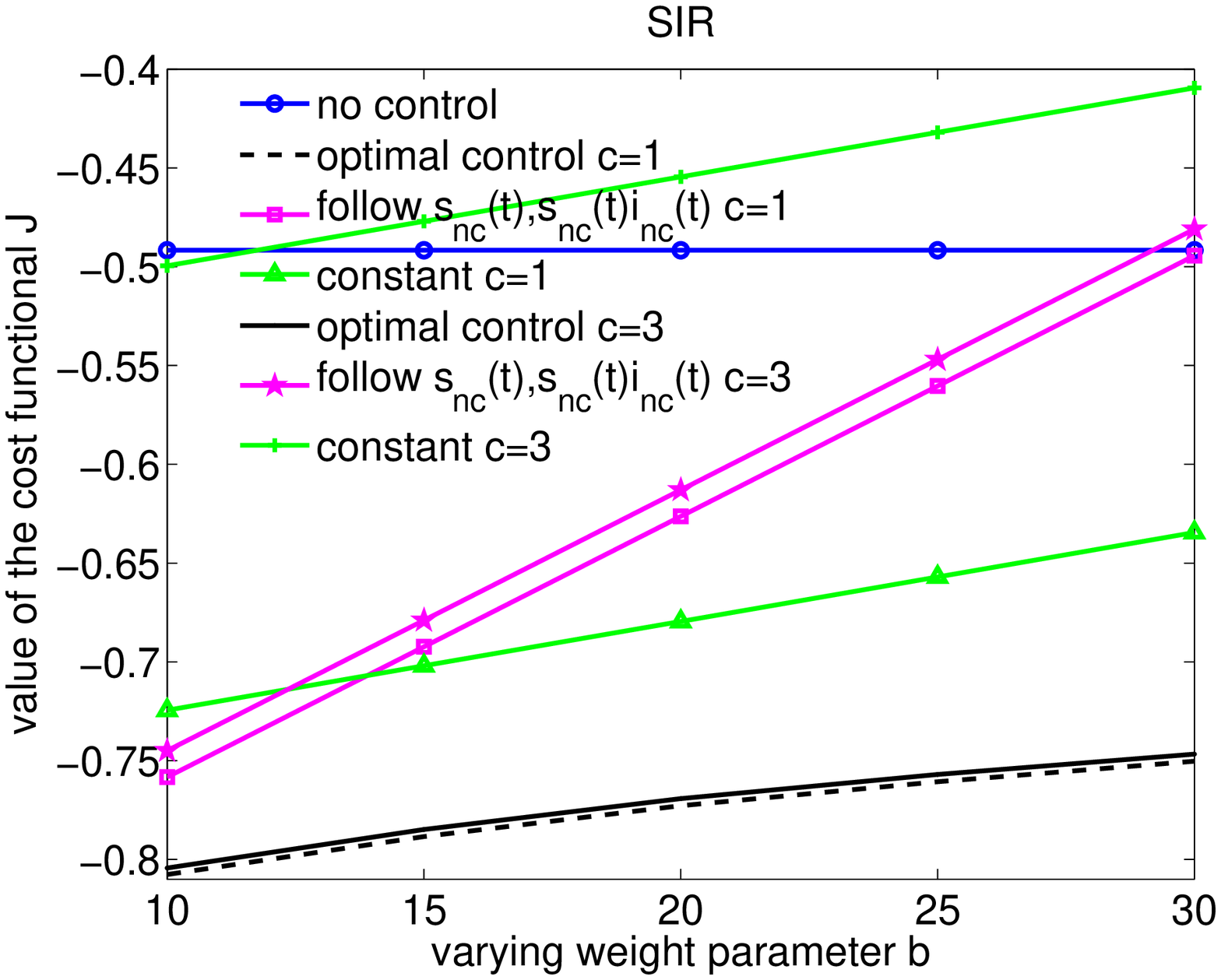} }
\caption{Cost functional $J$ vs weight parameters ($b, c$)}
\label{fig:J_vs_b_c}
\end{figure}

\subsection{Comparison Between Optimal, Constant and a Heuristic Control}
\label{sec:results-J_vs_parameters}
The effective spreading rate is again constant in this section, $\beta(t)=\beta,~\forall t\in[0,T]$. The aim of this section is to quantify the effectiveness of the optimal control strategy over simple and intuitive or ``common sense'' control strategies which do not involve any optimization.

We first introduce a simple heuristic control strategy which requires no knowledge of optimal control theory. Let $s_{nc}(t), i_{nc}(t)$ and $r_{nc}(t)$ be the fractions of susceptible, infected and recovered individuals at time $t$ when \emph{no control} is applied. Since the direct control targets susceptibles, in both the models, its effectiveness depends on the proportion of the susceptible population at the time it is applied. A reasonable heuristic direct control signal could be $\big(u_{1max}\cdot s_{nc}(t)\big)$ where $u_{1max}$ is the maximum allowed direct control in the given model. It adjusts the strength of the direct control signal according to the fraction of susceptibles in the no control scenario at any time instant. Word-of-mouth control requires infected individuals to convince susceptibles; hence, it is effective when the numbers of both susceptibles and infected individuals are significant. A heuristic word-of-mouth control signal for the model in Section \ref{sec:SIR_sys_model} could be $\big(u_{2max}\cdot s_{nc}(t)\cdot i_{nc}(t)\big)$, where $u_{2max}$ is the maximum allowed word-of-mouth control. We name these controls `follow $s_{nc}(t)$' for the SIS model and `follow $s_{nc}(t),s_{nc}(t)i_{nc}(t)$' for the SIR model. Please note that these controls are decided and fixed at the beginning of the campaign period (open loop strategies) and are obtained by referring to the quantities for the uncontrolled system. Another simple control strategy applies constant control throughout the campaign period; the control signal is set to half of the maximum allowed signal strength. Thus constant control has values $\frac{1}{2}u_{max}$ in the SIS model and $\frac{1}{2}u_{1max}$ (direct), $\frac{1}{2}u_{2max}$ (word-of-mouth) in the SIR model. 

Fig. \ref{fig:how_different_controls_and_states_look} shows the shapes of different control signals and the corresponding state evolutions for the SIS model. The cost functional involves weighted sums of the total control effort (area under the control curve) and the final fraction of infected individuals, $i(T)$. Notice that $i(T)$ is similar for the three strategies; however, the total control efforts are considerably different. For a better idea of the performance of the optimal control compared to the other strategies, we plot the cost functional $J$ with respect to one of the parameters $\beta,\gamma,T,b$ or $c$ for various strategies for both SIS and SIR models in Figs. \ref{fig:J_vs_beta}--\ref{fig:J_vs_b_c}.

We make following observations from the plots in this section:
\begin{enumerate}[(i)]
\item From Figs. \ref{fig:J_vs_beta}--\ref{fig:J_vs_b_c}: for some parameter vectors, the cost functional $J$ for the constant control strategy and `follow $s_{nc}(t)$' (for SIS model) or `follow $s_{nc}(t),s_{nc}(t)i_{nc}(t)$' (for SIR model) strategy is more than $J$ for no control. However, $J$ for the optimal control strategy is always smaller than $J$ for no control. This shows the advantage of optimal control over heuristic controls and illustrates the fact that (for some parameter vectors), an ill-planned campaign may prove more costly than no campaign at all.

\item With reference to Figs. \ref{fig:J_vs_beta} and \ref{fig:J_vs_Tdeadline} (for both the models), from the difference between the curves corresponding to no control and optimal control cases (for large values of $\beta$ and $T$ respectively), one can conclude the following: If the uncontrolled system is capable of achieving a high value of $(1-s(t))$ (either due to high $\beta$ or high $T$, given fixed values of other parameters), application of an optimal control strategy does not improve (decrease) costs too much compared to no control. Note that the cost of applying control is zero for `no control' case and $J$ in that case is nothing but $-(1-s(T))$. In other cases, application of the optimal control decreases the value of $J$ compared to the no control strategy.

\item From Figs. \ref{fig:J_vs_beta} and \ref{fig:J_vs_gamma}: As $R_0=\beta/\gamma$ increases, $J$ corresponding to the optimal control strategy decreases. In other words, more viral campaigns are less costly to run than less viral campaigns.

\item From Fig. \ref{fig:J_vs_b_c}: $J$ increases as one of the weight parameters $b$ or $c$ increases (other parameters held fixed). We incur more costs as application of control becomes dearer.

\item From Fig. \ref{fig:J_vs_b_c}: The relative increase in $J$ with respect to increase in $b$ or $c$ for the optimal control strategy is less than that for the `follow $s_{nc}(t)$' or the `follow $s_{nc}(t),s_{nc}(t)i_{nc}(t)$' strategy and constant control strategy. Thus, the optimal control strategy is less sensitive to changes in $b$ or $c$ than other control strategies.
\end{enumerate}

\section{Conclusion}
\label{sec:Conclusions}
In this paper, we have studied optimal control strategies for running campaigns on a homogeneously mixed population when the information spreading rate is a function of time. The change in the spreading rate over time reflects the change in the interest level of the population in the subject of the campaign. The first model assumes that information spreads through an SIS process and the campaigner can directly recruit members of the population, at some cost, to act as spreaders. The second model allows the campaigner to incentivize infected individuals (leading to increased effective spreading rate), in addition to the direct recruitment in the SIR epidemic process. We have shown the existence of solutions for the two models, and uniqueness of the solutions for sufficiently small campaign deadline. For both the cases, for constant spreading rate, we have showed the effectiveness of the optimal control strategy over the constant control strategy, a heuristic control strategy and no control. We have shown the sensitivity of control to the time varying spreading rate profile. Our study can provide useful insights to campaign managers working to disseminate a piece of information in the most cost effective manner. Estimating parameter values such as information spreading rate and recovery rate are nontrivial for real world campaigns, and studying sensitivity of control strategies to estimation errors forms an interesting future research direction.

\appendix

\section{Proof for Theorem \ref{theorem:SIS_uniqueness} (Uniqueness of the Solution to the SIS Model)}
\label{app:uniqueness_SIS}

For sufficiently small deadline, $T$, the uniqueness of the solution to the SIS model can be established using techniques similar to \citep{fister1998optimizing}. If the solution to the optimal control problem is non-unique, consider two solutions $(i, \lambda)$ and $(\hat{i}, \hat{\lambda})$. The time variable $t$ is dropped for notational brevity. Without loss of generality, for $0\leq t\leq T$, let, 
\begin{equation}
i = e^{a t}x,~~ \lambda = e^{- a t}y,~~ \hat{i} = e^{a t}\hat{x} \text{~~and~~} \hat{\lambda} = e^{- a t}\hat{y}, \label{eq:uniq_subs_SIS}
\end{equation} 
where $a$ is a positive real number. Note that $x,y,\hat{x},\hat{y}$ are functions of $t$, $0 \leq x, \hat{x} \leq 1$ and $0\leq y, \hat{y} \leq y_{max}$, for $0\leq t\leq T$. From (\ref{eq:SIS_u}),
\begin{eqnarray}
u &=& \text{min}\left\{\text{max}\left\{\frac{e^{-at}y(1-e^{at}x)}{2b},0\right\},u_{max}\right\}, \nonumber \\
\hat{u} &=& \text{min}\left\{\text{max}\left\{\frac{e^{-at}\hat{y}(1-e^{at}\hat{x})}{2b},0\right\},u_{max}\right\}. \nonumber
\end{eqnarray}
Thus we can estimate $(u-\hat u)^2$ as,
\begin{align}
& (u-\hat u)^2 \nonumber\\
\leq & \left(\frac{e^{-at}}{2b}(y-\hat y) - \frac{1}{2b}(xy-\hat x \hat y)\right)^2\nonumber\\
= & \Bigg(\left(\frac{e^{-at}-\hat x}{2b}\right)^2(y-\hat y)^2+\frac{y}{2b}(x-\hat x)^2 -2\frac{(e^{-at}-\hat x)y}{4b^2}(x-\hat x)(y-\hat y)\Bigg) \nonumber \\
\leq & \Bigg(\left(\frac{e^{-at}-\hat x}{2b}\right)^2(y-\hat y)^2+\frac{y}{2b}(x-\hat x)^2+\left|\frac{(e^{-at}-\hat x)y}{2b^2}\right||(x-\hat x)(y-\hat y)|\Bigg) \nonumber \\
\leq & A_1(x-\hat x)^2 +A_2(y-\hat y)^2 \label{eq:u_ucap_sq}
\end{align}
We have used $(xy-\hat x \hat y)=(xy-\hat xy+\hat xy-\hat x \hat y)$ and $m^2+n^2\geq 2|mn|$.

Using (\ref{eq:uniq_subs_SIS}) in (\ref{eq:SIS_constraint})
\begin{equation}
e^{at} \dot{x} + a x e^{at} = - \beta(t) e^{2at} x^2 + \big(\beta(t) - \gamma - u\big)e^{at} x + u. \nonumber
\end{equation}
Writing similar equation for $\frac{d\hat{i}}{dt}$ and subtracting from above we get,
\begin{align}
& e^{at} \left(\dot{x}-\dot{\hat x}\right) + a e^{at} (x - \hat x) \nonumber \\
= & - \beta(t) e^{2at} (x^2-\hat x^2) + \big(\beta(t) - \gamma \big)e^{at} (x-\hat x) - e^{at}(ux-\hat u \hat x)+ (u-\hat u). \nonumber
\end{align}
Multiplying both sides by $(x-\hat x)$
\begin{align}
& e^{at}(x-\hat x)(\dot{x}-\dot{\hat x}) + a e^{at} (x - \hat x)^2 \nonumber \\
= & -\beta(t)e^{2at}(x+\hat x)(x-\hat x)^2 + e^{at}\big(\beta(t) - \gamma \big)(x-\hat x)^2 - e^{at}u(x-\hat x)^2+(1-e^{at}\hat x)(u-\hat u)(x-\hat x)\nonumber \\
\leq & e^{aT}C_1(x-\hat x)^2 + C_2|(u-\hat u)(x-\hat x)| \nonumber\\
\leq & e^{aT}C_3(x-\hat x)^2 + C_4(u-\hat u)^2. \nonumber
\end{align}
$C_2=\text{max}(|1-e^{at}\hat x|)$. Integrating both sides with respect to $t$ from $0$ to $T$, we get,
\small
\begin{equation}
\frac{1}{2}(x-\hat x)^2(T) + (a-C_3e^{aT}) \int_{0}^{T} (x-\hat x)^2 dt \leq C_4 \int_{0}^{T} (u-\hat u)^2 dt \label{eq:SIS_uniq_state}
\end{equation}
\normalsize

Substituting (\ref{eq:uniq_subs_SIS}) in (\ref{eq:SIS_costate}) we get two equations for $i,\lambda$ and $\hat i, \hat \lambda$, subtracting them we get,
\begin{align}
&(\dot{y}-\dot{\hat y})-a(y-\hat y)\nonumber \\
=& 2\beta(t)e^{at}(xy-\hat x \hat y)-(\beta(t)-\gamma)(y-\hat y)+(uy-\hat u \hat y). \nonumber
\end{align}
Multiplying both sides by $y-\hat y$ we get,
\begin{align}
&(y-\hat y)(\dot{y}-\dot{\hat y})-a(y-\hat y)^2 \nonumber \\
=& 2\beta(t)e^{at}x(y-\hat y)^2 +2\beta(t)e^{at}\hat y(x-\hat x)(y-\hat y) -(\beta(t)-\gamma)(y-\hat y)^2+y(u-\hat u)(y-\hat y) + \hat u(y-\hat y)^2. \nonumber
\end{align}
Integrating both sides with respect to $t$ from $0$ to $T$, we get,
\small
\begin{align}
& \frac{1}{2}(y-\hat y)^2(0)+a\int_0^T(y-\hat y)^2 \nonumber \\
=& -\int_0^T 2\beta(t)e^{at}x(y-\hat y)^2 dt -\int_0^T2 \beta(t) e^{at}\hat y (x-\hat x)(y-\hat y)dt + \int_0^T(\beta(t)-\gamma)(y-\hat y)^2 dt - \int_0^T y(u-\hat u)(y-\hat y)\nonumber \\
& -\int_0^T\hat u(y-\hat y)^2 \nonumber \\
\leq & \int_0^T|2\beta(t)e^{at}\hat y||(x-\hat x)(y-\hat y)|dt+(\beta_{max}-\gamma)\int_0^T(y-\hat y)^2 dt + \int_0^T|y||(u-\hat u)(y-\hat y)| dt.\nonumber
\end{align}
\normalsize
Which leads to,
\small
\begin{align}
& \frac{1}{2}(y-\hat y)^2(0)+(a-C_5e^{aT})\int_0^T(x-\hat x)^2 + (a-C_6e^{aT}-C_7)\int_0^T(y-\hat y)^2 
\leq C_8\int_{0}^{T} (u-\hat u)^2 dt \label{eq:SIS_uniq_costate}.
\end{align}
\normalsize
Substituting (\ref{eq:u_ucap_sq}) in (\ref{eq:SIS_uniq_state})+(\ref{eq:SIS_uniq_costate}) we get,
\begin{align}
& \frac{1}{2}(y-\hat y)^2(T)+\frac{1}{2}(y-\hat y)^2(0) + (a-C_8e^{aT}-C_9)\int_0^T(x-\hat x)^2 + (a-C_{10}e^{aT}-C_{11})\int_0^T(y-\hat y)^2 \leq 0, \label{eq:SIS_uniq_final}
\end{align}
which leads to the conclusion that $x=\hat x$ and $y=\hat y$ for,
\begin{equation}
T\leq \text{inf}\left\{\underset{a}{\text{sup}}\left\{\frac{1}{a}\log_e\left(\frac{a-C_9}{C_8}\right)\right\},\underset{a}{\text{sup}}\left\{\frac{1}{a}\log_e\left(\frac{a-C_{11}}{C_{10}}\right)\right\}\right\} \nonumber
\end{equation}
Thus the solution to state and costate equations and hence the optimal control is unique for sufficiently small deadline, $T$. Notice that $C_i >0,~i=8,...,11$; otherwise it can be estimated as $0$ due to inequality in (\ref{eq:SIS_uniq_final}).

\section{Proof for Theorem \ref{theorem:SIR_wom_uniqueness} (Uniqueness of the Solution to the SIR Model with Direct Recruitment and Word-of-mouth Control)}
\label{app:uniqueness_SIR_wom}
For sufficiently small deadline, $T$, the uniqueness of the solution can be shown. If the solution to the optimal control problem is non-unique, consider the two solutions $(r,s,\lambda_s,\lambda_r)$ and $(\hat r,\hat s,\hat \lambda_s,\hat \lambda_r)$. The time variable $t$ is dropped for notational brevity. Without loss of generality, for $0\leq t\leq T$, let, $s = e^{a t}x,r=e^{at}y,\lambda_s = e^{- a t}p,\lambda_r=e^{-at}q$ and $\hat{s} = e^{a t}\hat{x}, \hat{r} = e^{a t}\hat{y}, \hat{\lambda_s} = e^{- a t}\hat{p}, \hat{\lambda_r} = e^{- a t}\hat{q}$. The technique used is same as in \ref{app:uniqueness_SIS}, hence only the final form of the estimations are shown here.
\begin{align}
(u_1-\hat u_1)^2 & \leq A_1(x-\hat x)^2 +A_2(p-\hat p)^2, \nonumber \\
(u_2-\hat u_2)^2 & \leq (A_3e^{aT}+A_4)(x-\hat x)^2 + A_5e^{aT}(y-\hat y)^2 +(A_6e^{aT}+A_7)(p-\hat p)^2. \nonumber
\end{align}

From the state equations (\ref{eq:SIR_modified_constraint1}) and (\ref{eq:SIR_modified_constraint2}) we get,
\small
\begin{align}
\frac{1}{2}(x-\hat x)^2(T) + (a-C_1e^{aT}) \int_{0}^{T} (x-\hat x)^2 dt + (-C_2e^{aT}) \int_{0}^{T} (y-\hat y)^2 dt \nonumber \\
\leq  C_3 \int_{0}^{T} (u_1-\hat u_1)^2 dt + (C_4e^{aT}+C_5) \int_{0}^{T} (u_2-\hat u_2)^2 dt \nonumber
\end{align}
and,
\begin{align}
\frac{1}{2}(y-\hat y)^2(T) + (-C_6) \int_{0}^{T} (x-\hat x)^2 dt + (a-C_7) \int_{0}^{T} (y-\hat y)^2 dt \leq 0. \nonumber
\end{align}
\normalsize

Costate equations (\ref{eq:SIR_wom_lam_s}) and (\ref{eq:SIR_wom_lam_r}) lead to,
\small
\begin{align}
\frac{1}{2}(p-\hat p)^2(0) + (-C_8e^{aT}) \int_{0}^{T} (x-\hat x)^2 dt + (-C_9e^{aT}) \int_{0}^{T} (y-\hat y)^2 dt  + (a-C_{10}e^{aT}-C_{11}) \int_{0}^{T} (p-\hat p)^2 dt \nonumber \\
+ (-C_{12}) \int_0^T(q-\hat q)^2 dt \leq  C_{13} \int_{0}^{T} (u_1-\hat u_1)^2 dt + (C_{14}e^{aT}+C_{15}) \int_{0}^{T} (u_2-\hat u_2)^2 dt \nonumber
\end{align}
\normalsize
and,
\small
\begin{align}
\frac{1}{2}(q-\hat q)^2(0) + (-C_{16}e^{aT}) \int_{0}^{T} (x-\hat x)^2 dt  + (-C_{17}e^{aT}) \int_{0}^{T} (p-\hat p)^2 dt + (a-C_{18}e^{aT}-C_{19}) \int_0^T(q-\hat q)^2 dt \nonumber \\
\leq (C_{20}e^{aT}) \int_{0}^{T} (u_2-\hat u_2)^2 dt \nonumber
\end{align}
\normalsize

Finally the estimates from state and costate equations are added and the estimates of $(u_1-\hat u_1)^2$ and $(u_1-\hat u_1)^2$ are used in the right hand side to get an inequality of the form,
\small
\begin{align}
& \frac{1}{2}(x-\hat x)^2(T) + \frac{1}{2}(y-\hat y)^2(T) + \frac{1}{2}(p-\hat p)^2(0) + \frac{1}{2}(q-\hat q)^2(0)  + (a-D_1e^{aT}-D_2) \int_{0}^{T} (x-\hat x)^2 dt \nonumber \\
& + (a-D_3e^{aT}-D_4) \int_{0}^{T} (y-\hat y)^2 dt + (a-D_5e^{aT}-D_6) \int_{0}^{T} (p-\hat p)^2 dt + (a-D_7e^{aT}-D_8) \int_0^T(q-\hat q)^2 dt \leq 0. \nonumber
\end{align}
\normalsize
Notice that $D_i>0,~i=1,...,8$ otherwise it can be estimated to $0$ due to above inequality. Finally, $x=\hat x,y=\hat y,p=\hat p,q=\hat q$ for,
\begin{eqnarray}
T\leq \text{inf}\Bigg\{\underset{a}{\text{sup}}\left\{\frac{1}{a}\log_e\left(\frac{a-D_2}{D_1}\right)\right\},\underset{a}{\text{sup}}\left\{\frac{1}{a}\log_e\left(\frac{a-D_4}{D_3}\right)\right\}, \underset{a}{\text{sup}}\left\{\frac{1}{a}\log_e\left(\frac{a-D_6}{D_5}\right)\right\},\underset{a}{\text{sup}}\left\{\frac{1}{a}\log_e\left(\frac{a-D_8}{D_7}\right)\right\}\Bigg\}. \nonumber
\end{eqnarray}
Thus the solution is unique for sufficiently small deadline, $T$.

\bibliographystyle{model1-num-names}

\bibliography{bibliography_database}

\begin{thebibliography}{24}
\expandafter\ifx\csname natexlab\endcsname\relax\def\natexlab#1{#1}\fi
\providecommand{\bibinfo}[2]{#2}
\ifx\xfnm\relax \def\xfnm[#1]{\unskip,\space#1}\fi
\bibitem[{Asano et~al.(2008)Asano, Gross, Lenhart, and Real}]{asano2008optimal}
\bibinfo{author}{E.~Asano}, \bibinfo{author}{L.~J. Gross},
  \bibinfo{author}{S.~Lenhart}, \bibinfo{author}{L.~A. Real},
\newblock \bibinfo{title}{{Optimal Control of Vaccine Distribution in a Rabies
  Metapopulation Model}},
\newblock \bibinfo{journal}{Mathematical Bioscience and Engineering}
  \bibinfo{volume}{5} (\bibinfo{year}{2008}) \bibinfo{pages}{219--238}.
\bibitem[{Behncke(2000)}]{behncke2000}
\bibinfo{author}{H.~Behncke},
\newblock \bibinfo{title}{{Optimal Control of Deterministic Epidemics}},
\newblock \bibinfo{journal}{Optimal Control Applications and Methods}
  \bibinfo{volume}{21} (\bibinfo{year}{2000}) \bibinfo{pages}{269--285}.
\bibitem[{Gaff and Schaefer(2009)}]{gaff2009optimal}
\bibinfo{author}{H.~Gaff}, \bibinfo{author}{E.~Schaefer},
\newblock \bibinfo{title}{{Optimal Control Applied to Vaccination and Treatment
  Strategies for Various Epidemiological Models}},
\newblock \bibinfo{journal}{Mathematical Bioscience and Engineering}
  \bibinfo{volume}{6} (\bibinfo{year}{2009}) \bibinfo{pages}{469}.
\bibitem[{Lashari and Zaman(2012)}]{lashari2012optimal}
\bibinfo{author}{A.~A. Lashari}, \bibinfo{author}{G.~Zaman},
\newblock \bibinfo{title}{{Optimal Control of a Vector Borne Disease with
  Horizontal Transmission}},
\newblock \bibinfo{journal}{Nonlinear Analysis: Real World Applications}
  \bibinfo{volume}{13} (\bibinfo{year}{2012}) \bibinfo{pages}{203--212}.
\bibitem[{Ledzewicz and Sch{\"a}ttler(2011)}]{ledzewicz2011optimal}
\bibinfo{author}{U.~Ledzewicz}, \bibinfo{author}{H.~Sch{\"a}ttler},
\newblock \bibinfo{title}{{On Optimal Singular Controls for a General SIR Model
  With Vaccination and Treatment}},
\newblock \bibinfo{journal}{Discrete and Continuous Dynamical Systems}
  (\bibinfo{year}{2011}) \bibinfo{pages}{981--990}.
\bibitem[{Morton and Wickwire(1974)}]{morton1974}
\bibinfo{author}{R.~Morton}, \bibinfo{author}{K.~H. Wickwire},
\newblock \bibinfo{title}{{On the Optimal Control of a Deterministic
  Epidemic}},
\newblock \bibinfo{journal}{Advances in Applied Probability}
  (\bibinfo{year}{1974}) \bibinfo{pages}{622--635}.
\bibitem[{Sethi and Staats(1978)}]{sethi1978}
\bibinfo{author}{S.~P. Sethi}, \bibinfo{author}{P.~W. Staats},
\newblock \bibinfo{title}{{Optimal Control of some Simple Deterministic
  Epidemic Models}},
\newblock \bibinfo{journal}{Journal of Operational Research Society}
  (\bibinfo{year}{1978}) \bibinfo{pages}{129--136}.
\bibitem[{Yan and Zou(2008)}]{yan2008}
\bibinfo{author}{X.~Yan}, \bibinfo{author}{Y.~Zou},
\newblock \bibinfo{title}{{Optimal Internet Worm Treatment Strategy Based on
  the Two-Factor Model}},
\newblock \bibinfo{journal}{Electronics and Telecommunications Research
  Institute Journal} \bibinfo{volume}{30} (\bibinfo{year}{2008}).
\bibitem[{Zhu et~al.(2012)Zhu, Yang, Yang, and Zhang}]{zhu2012optimal}
\bibinfo{author}{Q.~Zhu}, \bibinfo{author}{X.~Yang}, \bibinfo{author}{L.~X.
  Yang}, \bibinfo{author}{C.~Zhang},
\newblock \bibinfo{title}{{Optimal Control of Computer Virus Under a Delayed
  Model}},
\newblock \bibinfo{journal}{Applied Mathematics and Computation}
  (\bibinfo{year}{2012}).
\bibitem[{Castilho(2006)}]{castilho2006optimal}
\bibinfo{author}{C.~Castilho},
\newblock \bibinfo{title}{{Optimal Control of an Epidemic Through Educational
  Campaigns}},
\newblock \bibinfo{journal}{Electronic Journal of Differential Equations}
  (\bibinfo{year}{2006}) \bibinfo{pages}{1--11}.
\bibitem[{Karnik and Dayama(2012)}]{karnik2012}
\bibinfo{author}{A.~Karnik}, \bibinfo{author}{P.~Dayama},
\newblock \bibinfo{title}{{Optimal Control of Information Epidemics}},
\newblock \bibinfo{journal}{Proceedings of IEEE Communication Systems and
  Netetworks Conference}  (\bibinfo{year}{2012}) \bibinfo{pages}{1--7}.
\bibitem[{Chierichetti et~al.(2009)Chierichetti, Lattanzi, and
  Panconesi}]{chierichetti2009}
\bibinfo{author}{F.~Chierichetti}, \bibinfo{author}{S.~Lattanzi},
  \bibinfo{author}{A.~Panconesi},
\newblock \bibinfo{title}{{Rumor Spreading in Social Networks}},
\newblock \bibinfo{journal}{Automata Languages and Programming}
  (\bibinfo{year}{2009}) \bibinfo{pages}{375--386}.
\bibitem[{Pittel(1987)}]{pittel1987spreading}
\bibinfo{author}{B.~Pittel},
\newblock \bibinfo{title}{{On Spreading a Rumor}},
\newblock \bibinfo{journal}{SIAM Journal on Applied Mathematics}
  (\bibinfo{year}{1987}) \bibinfo{pages}{213–--223}.
\bibitem[{Boyd et~al.(2005)Boyd, Ghosh, Prabhakar, and Shah}]{boyd2005}
\bibinfo{author}{S.~Boyd}, \bibinfo{author}{A.~Ghosh},
  \bibinfo{author}{B.~Prabhakar}, \bibinfo{author}{D.~Shah},
\newblock \bibinfo{title}{{Gossip Algorithms: Design, Analysis and
  Applications}},
\newblock \bibinfo{journal}{Proceedings of IEEE International Conference on
  Computer Communications}  (\bibinfo{year}{2005}) \bibinfo{pages}{1653--1664}.
\bibitem[{Belen(2008)}]{belen2008}
\bibinfo{author}{S.~Belen},
\newblock \bibinfo{title}{{The Behaviour of Stochastic Rumours}},
\newblock \bibinfo{journal}{Ph.D. dissertation, University of Adelaide,
  Australia.}  (\bibinfo{year}{2008}).
\bibitem[{Sethi et~al.(2008)Sethi, A., and He}]{sethi2008optimal}
\bibinfo{author}{S.~P. Sethi}, \bibinfo{author}{P.~A.},
  \bibinfo{author}{X.~He},
\newblock \bibinfo{title}{{Optimal Advertising and Pricing in a New-Product
  Adoption Model}},
\newblock \bibinfo{journal}{Journal of Optimization Theory and Applications}
  (\bibinfo{year}{2008}) \bibinfo{pages}{351--360}.
\bibitem[{Krishnamoorthy et~al.(2010)Krishnamoorthy, Prasad, and
  Sethi}]{krishnamoorthy2010optimal}
\bibinfo{author}{A.~Krishnamoorthy}, \bibinfo{author}{A.~Prasad},
  \bibinfo{author}{S.~P. Sethi},
\newblock \bibinfo{title}{{Optimal Pricing and Advertising in a Durable-Good
  Duopoly}},
\newblock \bibinfo{journal}{European Journal of Operations Research}
  (\bibinfo{year}{2010}) \bibinfo{pages}{486--497}.
\bibitem[{Khouzani et~al.(2011)Khouzani, Sarkar, and
  Altman}]{khouzani2011optimal}
\bibinfo{author}{M.~H.~R. Khouzani}, \bibinfo{author}{S.~Sarkar},
  \bibinfo{author}{E.~Altman},
\newblock \bibinfo{title}{{Optimal Control of Epidemic Evolution}},
\newblock \bibinfo{journal}{Proceedings of IEEE International Conference on
  Computer Communications}  (\bibinfo{year}{2011}) \bibinfo{pages}{1683--1691}.
\bibitem[{Barrat et~al.(2008)Barrat, Barthlemy, and
  Vespignani}]{barrat2008dynamical}
\bibinfo{author}{A.~Barrat}, \bibinfo{author}{M.~Barthlemy},
  \bibinfo{author}{A.~Vespignani}, \bibinfo{title}{{Dynamical Processes on
  Complex Networks}}, \bibinfo{publisher}{Cambridge University Press},
  \bibinfo{year}{2008}.
\bibitem[{Fleming and Rishel(1975)}]{fleming1975deterministic}
\bibinfo{author}{W.~H. Fleming}, \bibinfo{author}{R.~W. Rishel},
  \bibinfo{title}{{Deterministic and Stochastic Optimal Control}},
  \bibinfo{publisher}{Springer-Verlag, New York}, \bibinfo{year}{1975}.
\bibitem[{Birkhoff and Rota(1989)}]{birkhoff1989ordinary}
\bibinfo{author}{G.~Birkhoff}, \bibinfo{author}{G.~C. Rota},
  \bibinfo{title}{{Ordinary Differential Equations}}, \bibinfo{publisher}{John
  Wiley and Sons}, \bibinfo{year}{1989}.
\bibitem[{Kamien and Schwartz(1991)}]{kamien1991dynamic}
\bibinfo{author}{M.~I. Kamien}, \bibinfo{author}{N.~L. Schwartz},
  \bibinfo{title}{{Dynamic Optimization: the Calculus of Variations and Optimal
  Control in Economics and Management}}, \bibinfo{publisher}{North-Holland
  Amsterdam}, \bibinfo{year}{1991}.
\bibitem[{Kutz(2005)}]{kutz2005practical}
\bibinfo{author}{N.~J. Kutz}, \bibinfo{title}{{Practical Scientific Computing,
  AMATH 581 Course Notes}}, \bibinfo{publisher}{Department of Applied
  Mathematics, University of Washington}, \bibinfo{year}{2005}.
\bibitem[{Fister et~al.(1998)Fister, Lenhart, and
  McNally}]{fister1998optimizing}
\bibinfo{author}{K.~R. Fister}, \bibinfo{author}{S.~Lenhart},
  \bibinfo{author}{J.~S. McNally},
\newblock \bibinfo{title}{{Optimizing Chemotherapy in an HIV Model}},
\newblock \bibinfo{journal}{Electronic Journal of Differential Equations}
  (\bibinfo{year}{1998}) \bibinfo{pages}{1--12}.

\end{thebibliography}

\end{document}